\documentclass[preprint]{elsarticle}

\usepackage{color}
\usepackage{graphicx}
\usepackage{amsfonts}
\usepackage{amssymb}
\usepackage{latexsym}
\usepackage{amsmath}
\usepackage{amsthm}
\usepackage{mathrsfs}

\definecolor{purple}{rgb}{0.65, 0, 1}
\definecolor{orange}{rgb}{1,.5,0}

\newtheoremstyle{thm}{1.5ex}{1.5ex}{\itshape\rmfamily}{}
{\bfseries\rmfamily}{}{2ex}{}

\newtheoremstyle{rem}{1.3ex}{1.3ex}{\rmfamily}{}
{\itshape}
{} {1.5ex}{}

\theoremstyle{thm}

\newtheorem{thm}{Theorem}[section]

\newtheorem{corollary}[thm]{Corollary}
\newtheorem{lemma}[thm]{Lemma}
\newtheorem{proposition}[thm]{Proposition}

\theoremstyle{definition}

\newtheorem{theorem}{Theorem}[section]

\title{\bf Territorial Developments Based on Graffiti: a 
Statistical Mechanics Approach}

\author[case]{Alethea B. T. Barbaro}
\ead{alethea.barbaro@case.edu}
\author[ucla]{Lincoln Chayes}
\ead{lchayes@math.ucla.edu}
\author[csun]{Maria R. D'Orsogna \corref{cor1}}
\ead{dorsogna@csun.edu}
\cortext[cor1]{Corresponding author. Email address: dorsogna@csun.edu}

\address[ucla]{UCLA Mathematics Department\\
520 Portola Plaza \\
Box 951555\\
Los Angeles, CA 90095-1555\\
USA}

\address[csun]{
CSUN Mathematics Department\\ 
18111 Nordhoff St\\ 
Los Angeles, CA 91330-8313\\
USA 
}

\address[case]{
CWRU Department of Mathematics\\
10900 Euclid Avenue-Yost Hall Room 220\\
Cleveland, Ohio 44106-7058\\
USA
}

\begin{document}

\begin{abstract}
We study the well-known sociological phenomenon of gang aggregation
and territory formation through an interacting agent system defined on a
lattice. We introduce a two-gang Hamiltonian model where agents have
red or blue
affiliation but are otherwise
indistinguishable.
In this model, all interactions are indirect and occur only via graffiti
markings, on-site as well as on nearest neighbor locations.  We also
allow for gang proliferation and graffiti suppression.  Within the
context of this model, we show that gang clustering and territory
formation may arise under specific parameter choices and that a phase
transition may occur between well--mixed, possibly dilute
configurations and well separated, clustered ones.  Using methods from
statistical mechanics, we study the phase transition between these two
qualitatively different scenarios.
In the mean--fields rendition of this model, we identify parameter regimes
where the transition is first or second order. 
In all cases, we have found that the transitions are a consequence
solely of the gang to graffiti couplings, implying that direct gang to gang
interactions are not strictly necessary for gang territory formation;
in particular, graffiti may be the sole driving force behind gang
clustering.  We further discuss possible sociological -- as well as ecological --
ramifications of our results.

\end{abstract}

\maketitle


\section{Introduction} \label{S:intro}

\noindent
Lattice models have been extensively used in
the physical sciences over the past decades
to describe a wide variety of condensed matter
equilibrium and non equilibrium phenomena
see e.g., the reviews in
 \cite{baxter, wu, alea}.
Magnetization was the original application, but the list has grown to
include structural transitions in DNA \cite{DNA, DNA1, DNA2}, polymer
coiling \cite{polymer, polymer1}, cellular automata \cite{cellauto,
  cellauto1}, and gene regulation \cite{gene, gene1, gene2} to name a
few.  The resulting models are certainly simplified, but what they
lack in detail is compensated by their amenability to analytical
and computational treatment -- and, occasionally, to exact solution.  
Moreover, at least for the behavior in the vicinity of a continuous transition, 
the simplifications inherent in these approximate models may be presumed to be inconsequential.
In short, lattice models have proved extremely useful in the context of the physical, biological and even chemical sciences.
In more recent years, lattice models have also been applied to study
social phenomena \cite{castellano, stauffer1, stauffer}, such as
racial segregation \cite{race1, race2}, voter preferences \cite{voter,
  voter1, voter2}, opinion formation in financial markets
\cite{opinion, opinion1, opinion2}, and language changes in society
\cite{language1, language2, language3}, offering insight into
socioeconomic dynamics and equilibria.
In this paper we consider the problem of gang aggregation via
graffiti in what is -- to the best of our knowledge -- the first
application of lattice model results 
 to the emergence
of gang territoriality.  

Scratching words or painting images on visible surfaces is certainly
not a new phenomenon. Wall scribblings have survived from ancient
times and have been used to reconstruct historical events and to
understand societal attitudes and values. Today, graffiti (from the
Italian \emph{graffiare}, to scratch) is a pervasive characteristic of
all metropolitan areas \cite{Alonso2}. Several types of graffiti
exist. Some are political in nature, expressing activist views against
the current establishment; others are expressive or offensive
manifestations on love, sex or race. At times, the graffiti is a mark
of one's passage through a certain area, with prestige being
attributed to the most prolific or creative tagger or to one who is able to
reach inaccessible locations. The mark can be anything from a simple
signature to a more elaborate decorative aerosol painting
\cite{Phillips, Alonso}.  All of these types of graffiti are usually
scattered around the urban landscape and do not 
appear to
follow any
predetermined spatio--temporal pattern of evolution. They affect the
quality of life simply as random defacement of property, although
sometimes they are considered art \cite{Knox}.

On the other hand, \textit{gang} graffiti represents a much more serious
threat to the public, since it is usually a sign of the presence of
criminal gangs engaged in illegal or underground activities such as
drug trafficking or extortion \cite{Smith, Fagan}. Street gangs are
extremely territorial, and aim to preserve economic interests and
spheres of influence within the neighborhoods they control. 
A gang's ``turf'' is usually marked in a characteristic
style, recognizable to members and antagonists \cite{Brown, LC1974}
with incursions by enemies often resulting in violent acts.
The established boundaries between different gang factions are
sometimes respected peacefully, but more often become contested
locations where it is not uncommon for murders and assaults to occur
\cite{BB1993}. It is here, on the boundaries between gang turfs, that
the most intense graffiti activity is usually concentrated.

Several criminological and geographical studies have been presented
connecting gang graffiti and territoriality in American cities
\cite{Knox, Alonso, LC1974}. In particular, it is now considered
well-established that the spatial extent of a gang's area of influence
is strongly correlated to the spatial extent of that particular gang's
graffiti style or language. Furthermore, it is known that the
incidence of gang graffiti may change in time, reflecting specific
occurrences or neighborhood changes. For example, rival gangs may
alternate between periods of truce and hostility, the latter being
triggered by arrests or shootings. Similarly, boundaries may shift
locations when the racial or socio--economic makeup of a neighborhood
changes, creating new tensions, or when gang members migrate to new
communities \cite{Alonso2}. In all these cases, periods of more
intense gang hostility are usually accompanied by intense graffiti
marking and erasing by rival factions in contested or newly settled
boundary zones \cite{LC1974}.

The purpose of this paper is to present a mathematical model that
includes relevant sociological and geographical information relating
gang graffiti to gang activity. In particular, we study the
segregation of individuals into well defined gang clusters as driven
by gang graffiti, and the creation of boundaries between rival
gangs. We use a spin system akin to a 2D lattice Ising model to
formulate our problem through the language of statistical mechanics.
In this context, the site variables $s_{i}$ have two constituents
which represent `gang' and `graffiti' types, respectively, and
\textit{phase separation} is assumed to be the proxy for gang
clustering.  For the purpose of simplicity, we consider only two
gangs, hereafter referred to as the red and blue gang, whose members
we refer to as \textit{agents}.  Lattice sites may be occupied by
agents of either color or be void.  Since gang members are assumed to
tag their territory with graffiti of their same color, we also assign
a graffiti index to each site representing the preponderance of red or
blue markings.

In particular, agents are attracted to sites with graffiti of their
same color, and avoid locations marked by their opponents. We
deliberately avoid including direct interactions between gang members,
so that ``ferromagnetic" type gang--gang attractions exist only insofar
as they are mediated by the graffiti. On one hand this is
mathematically interesting: in the broader context of physical
systems, interactions are often mediated but rarely are indirect
interactions the subject of mathematical analysis. On the other hand,
by excluding direct gang interactions, we can specifically focus on
the role of graffiti in gang dynamics and segregation.  Furthermore,
as will be later discussed, under certain conditions, gang--gang
couplings may be unimportant, and one of the primary conclusions of
this work is that they appear to be unnecessary to account for the
observed phenomena of gang segregation.  In any case, we informally
state without proof that all the results of this work also hold if
explicit agent--agent interactions are included.

We thus write $s_i = (\eta_i, g_i)$, representing the agents and
graffiti configuration at site $i$, respectively. The former component
$\eta_i$ is discrete allowing, for simplicity, at most one agent on
each site.  The latter $g_i$ is continuous and, in principle,
unbounded.  We let $\mathbf{s}$ denote a spin configuration on the
entire lattice, and in Section \ref{S:Hamiltonian}, propose a
Hamiltonian, $\mathscr{H}(\mathbf{s})$, to embody all relevant
sociological information.  Once $\mathscr H(\mathbf{s})$ has been
determined, the probability for the occurrence of a spin configuration
$\mathbf{s}$ on a finite connected lattice $\Lambda\subset \mathbb
Z^{2}$ is determined by the corresponding Gibbs distribution
$\mathbb{F}(\mathbf{s})$. Note that due to the choices made on the
range of the $\eta_i, g_i$ values, $\mathbb{F}(\mathbf{s})$ is
discrete in the $\eta$ variables and continuous in the $g$ ones. It is
given by

\begin{equation*}
\mathbb{F}(\mathbf{s}) = \frac{1}{\mathcal{Z}}
\exp(-\mathscr{H}(\mathbf{s})),
\end{equation*}

\noindent
where $\mathcal{Z}$ is the partition function for the finite lattice
$\Lambda$ formally provided by the expression

\begin{equation*}
\mathcal{Z} = \sum_{\mathbf{s} \in \mathbb{S}}
\exp(-\mathscr{H}(\mathbf{s})).
\end{equation*}

\noindent
Here, $\mathbb{S}$ denotes the set of all possible configurations on
$\Lambda$ and the summation symbol is understood to be a summation
over the discrete components and an integration over the continuous
ones.  As usual, we begin with a finite lattice and its associated
boundary conditions, and obtain infinite volume results by taking the
appropriate limits.
Using techniques from statistical mechanics, we prove that
our system undergoes a phase transition as the coupling parameters are
varied.  In the unconstrained ensemble, certain parameter choices lead
to predominance of either the red or blue gang, indicating that for
configurations where the red to blue gang ratio is fixed at unity,
a phase separation will occur.  Conversely, in other regions of parameter
space, there is no dominance of either gang type, indicating that the
two are well-mixed and/or dilute. In this work we will investigate under 
which conditions to expect phase separation or gang dilution.

Our paper is organized as follows: in Section \ref{S:Hamiltonian}, we
give details of the model and in Section
\ref{S:nearestNeighPhaseTransition}, we prove that a phase transition
exists as a function of the relevant parameters. Since information on the
location of \textit{all} transition points is, by necessity,
incomplete we consider an approximation in the form of a simplified
mean field version of our Hamiltonian and derive the corresponding
mean field equations in Section \ref{S:MFHamiltonian}. Here, we show
that the mean field Hamiltonian also exhibits a phase
transition and we further prove that the latter is continuous in one
specified region of parameter space and first order in
another. Finally, in Section \ref{S:discussion} we end with a
discussion of potential sociological and ecological 
implications of our results.

\section{The Hamiltonian} \label{S:Hamiltonian}

Let us define a spin system on a finite lattice $\Lambda \subset \mathbb Z^{2}$.
Here, the spin at each site $i \in \Lambda$ is
denoted by $s_i = (\eta_i, g_i)$ and, we reiterate,  $\eta_i$ denotes the \emph{agent
  spin} and $g_i$ represents the \emph{graffiti field}.  We allow the
agent spin to be in the set $\{0, \pm 1\}$; $\eta_i = -1$ if the agent
at site $i$ belongs to the blue gang, $\eta_i = 0$ if there is no
agent, and $\eta_i = +1$ if the agent is a red gang member.  The
graffiti field is in the set of real numbers: $g_i > 0$ indicates an
excess of red graffiti, $g_i < 0$ an excess of blue graffiti, and, in
either case, $|g_{i}|$ indicates the magnitude of the excess.  We now
introduce the formal Hamiltonian $\mathscr{H}(\mathbf{s})$

\begin{equation} \label{E:Hamiltonian}
 - \mathscr{H}(\mathbf{s}) = J \sum_{<i,j>} \eta_i g_j + K \sum_{i}
   \eta_i g_i + \alpha \sum_{i} \eta_i^2- \lambda \sum_{i} g_i^2,
 \end{equation}

\noindent
where $\mathbf{s}$ is a given configuration on the full $\Lambda$
lattice, $i$ and $j$ index its sites and $\sum_{\langle i,j\rangle}$
is the sum taken over every bond between nearest neighbor sites
belonging to $\Lambda$.  
We discuss the role of spins on the lattice boundary $\Lambda^c$ in
Proposition \ref{3point4} and following sections. The expression in
\ref{E:Hamiltonian} will be referred to as the GI--Hamiltonian
(graffiti interaction Hamiltonian) and its corresponding partition
function will be denoted by an unadorned ${\mathcal Z}$. Note that
since $\eta_i$ is either $0$ or $\pm 1$, $\eta_{i}^{2} = |\eta_{i}|$;
however, we choose to display the above form to leave open the
possibility of $\eta_{i} \in \mathbb{Z}$.  As discussed earlier, there
are no \textit{explicit} agent--agent interactions in this model;
indeed, the structure of the Hamiltonian assumes that gang members
interact with each other only via the graffiti tagging.  As a result,
occupation at site $i$ by a gang member is ``energetically'' favored
only if nearest-neighbor and on-site graffiti are predominantly of its
same color.  The two coupling constants, $J$ for nearest-neighbor
interactions and $K$ for on-site occupation, reflect this trend.  The
$\alpha \eta_i^2$ term represents the proclivity of a given site to be
occupied by agents regardless of color, implying that gang members
carry a strong tendency to occupy unclaimed turf if $\alpha \gg 1$,
while $\alpha \ll - 1$ represents a natural paucity of gangs
altogether.  Finally, we assume graffiti imbalance of either color to
be energetically unfavorable via the $-\lambda g_i^2$ term. This can
be interpreted as natural decay of graffiti due to the elements, or to
police or community intervention.  For purposes of stability,
$\lambda$ must be positive.  Although the interactions $J$, $K$ are
tacitly assumed to be positive, generalizations to negative values may
be possible, and a corresponding analysis may be undertaken given the
proper sociological interpretations.

\section{Phase transition in the GI--system}	
\label{S:nearestNeighPhaseTransition}
\subsection{Low  temperature phase}

\noindent
The basic strategy we follow to demonstrate an ordered, ``low
temperature" phase is a \textit{contour} argument, here 
illustrated: Suppose that
$\eta_i$, the agent spin at site $i$, differs from the agent spin
$\eta_j$ at a different site $j$.  The two agent spins can differ
either by color, representing two different gang affiliations, or by
occupation, where one site is occupied and the other is void.  
At the scale of nearest neighbors, each
edge in the lattice can be defined as either a \emph{coherent} bond,
where the adjoining lattice sites are occupied and their agent spins
are identical, or as an \emph{incoherent} bond if this condition does
not hold. Thus, explicitly, $(\eta_i, \eta_j) = (1,1)$ or 
$(-1,-1)$ are coherent, and all the other types are not.

Let us now consider any path on the lattice that joins sites $i$ and
$j$.  Since $i$ and $j$ have agent spins which are not identical, it
must be the case that on any path between $i$ and $j$, there is an
incoherent bond.  Furthermore, these incoherent bonds must form a
closed contour on the dual lattice that separates $i$ from $j$.  In
the following subsections, we derive a bound on the probability of any
such incoherent bonds and their aggregation into contours.  When these
probabilities are small enough -- which happens in certain regions of
the parameter space -- we can establish a low temperature
phase. For example, the presence of a red agent at
the origin will imply that, with significant probability, the majority of
the other sites will also be occupied by red agents, 
showing the existence of a red phase.  Similarly, a blue phase can be
shown to exist.

To achieve all of these ends, we will employ the methods of
\textit{reflection positivity} described in \cite{Biskup} and
\cite{SST} which contain a detailed account of useful techniques along
with relevant classic references.  In this paper, we will be working
on the $L\times L$ \textit{diagonal} 2D torus -- the SST -- which we
denote by $\mathbb T_{L}$. We will often refer to the Gibbsian
probability measure on $\mathbb T_{L}$ associated with the Hamiltonian
in Eq.(\ref{E:Hamiltonian}) which we denote by $\mathbb
P_{L}(\cdot)$.

\subsubsection{Reflection positivity}
By means of the reflection positivity of the Gibbs distribution we can
easily bound the expectation of an observable which depends only on
the spin at any two neighboring lattice points. This result will be
used to build the contour argument that will lead us to prove the
existence of a low temperature phase. We thus briefly introduce the
concept of reflection positivity, referring the interested reader to
\cite{Biskup} for a more detailed discussion of these topics.

Consider a plane of reflection $p$ which intersects the torus in a
path running through next nearest (diagonal) pairs of sites.  Let
$\vartheta_p$ be the reflection operator through $p$.  On the SST,
this plane $p$ divides the lattice into two halves, identified as
$\mathbb{T}^{+}_L$ and $\mathbb{T}^{-}_L$, such that
$\mathbb{T}^{+}_L\cap \mathbb{T}^{-}_L = p$.  Let $\mathscr U^{+}_{p}$
denote the set of functions which depend only on the spin variables in
$\mathbb{T}^{+}_L$ and similarly for $\mathscr U^{-}_{p}$.  The
\textit{reflection map}, $\vartheta_{p}$, which, in a natural fashion
identifies sites in $\mathbb{T}^{+}_L$ with those in
$\mathbb{T}^{-}_L$ via a reflection through $p$, can also be used to
define maps between $\mathscr U^{+}_{p}$ and $\mathscr U^{-}_{p}$:
Specifically, if $f\in \mathscr U^{+}_{p}$, we define $\vartheta_{p} f
\in \mathscr U^{-}_{p}$ to be the function $f$ evaluated on the
configuration reflected from $\mathbb{T}^{-}_L$.

A measure $\mu$ is
\emph{reflection positive} with respect to $\vartheta_p$ if for every $f,g \in \mathscr{U}^{+}_{p}$, or $\mathscr U^{-}_{p}$,
the following two properties hold

\begin{enumerate}
 \item $\mathbb{E}_{\mu} (f \vartheta_p f ) \geq 0$,
 \item $\mathbb{E}_{\mu} (f \vartheta_p g) = 
\mathbb{E}_{\mu}( g \vartheta_p f )$.
\end{enumerate}

\noindent
It is known (e.g., see \cite{Biskup}) that $\mathbb P_L$ is reflection
positive with respect to $\vartheta_p$ for every $p$ of the above
described type. We next use reflection positivity to find an upper
bound on the expectation of observables defined on bonds.  In doing so, we use the following lemmas:

\begin{lemma} \label{L:tilingBound}
Let $\langle i,j \rangle$ denote a bond of $\mathbb T_{L}$ and let
$\alpha_{i}$ and $\gamma_{j}$ denote site events at the respective
endpoints of the bond.  Let $\mathcal Z^{(\alpha,\gamma)}_{\mathbb
  T_{L}}$ denote the partition function (on $\mathbb T_{L}$) which has been
constrained so that at each site with the parity of $i$, the
translation of the event $\alpha_{i}$ occurs and similarly for
$\gamma$.  Then, for $L = 2^{k}$ for some integer $k$,
$$
\mathbb P_{L}(\alpha_{i}\cap\gamma_{j}) \leq
\left [
\frac{\mathcal Z^{(\alpha,\gamma)}_{\mathbb T_{L}}}{\mathcal Z_{\mathbb T_{L}}}
\right ]^{\frac{1}{2V}},
$$
where $V = L^{2}$ is the volume of the torus.
\end{lemma}
\begin{proof}
The result from this Lemma dates back to the original papers on the
subject. In particular, the use of bond events on the SST was
highlighted in \cite{SST}.  A modern and complete
derivation is contained in \cite{Biskup}, Section 5.3.
\end{proof}

\noindent
For a slightly more general scenario, let us 
consider the bond $\langle i,j\rangle$
and various events $\alpha_{i}^{1}, \gamma_{j}^{1}$, \dots ,
$\alpha_{i}^{n}, \gamma_{j}^{n}$ and let us denote by $b_{1} = \alpha_{i}^{1}
\cap \gamma_{j}^{1}$ \dots $b_{n} = \alpha_{i}^{n} \cap
\gamma_{j}^{n}$ the corresponding bond events as described.  Letting
$b = \cup_{j=1}^{n}b_j$ we find
$$
\mathbb P_{L}(b) \leq
\sum_{j=1}^{n}
\left [
\frac{\mathcal Z^{(\alpha_j,\gamma_{j})}_{\mathbb T_{L}}}{\mathcal Z_{\mathbb T_{L}}}
\right ]^{\frac{1}{2V}}
:=
\sum_{j=1}^{n}
\left [
\frac{\mathcal Z^{(b_{j})}_{\mathbb T_{L}}}{\mathcal Z_{\mathbb T_{L}}}
\right ]^{\frac{1}{2V}}.
$$

\noindent
Finally, we have
\begin{lemma}
\label{XDS}
Let $r_{1}, \dots r_{m}$ denote translations of the bond
$\langle i,j\rangle$
and $b_{r_{j}}$ the translation of the bond event(s) $b$ described above.  Then
$$
\mathbb P_{L}(\cap_{j = 1}^{m}b_{r_{j}})  
\leq
\left[
\sum_{j=1}^{n}
\left [
\frac{\mathcal Z^{(b_{j})}_{\mathbb T_{L}}}{\mathcal Z_{\mathbb T_{L}}}
\right ]^{\frac{1}{2V}}
\right ]^{m}.
$$
\begin{proof}
Again, we refer the reader to \cite{Biskup}, Section 5.3.
\end{proof}

\end{lemma}

\subsubsection{A bound on the incoherent bond probabilities} 
\label{SSS:boundingTheProb}

\noindent
In order to prove a phase transition by a contour argument, we must
place an upper bound on the probability for the occurrence of any type
of incoherent bond where agent spins of neighboring sites are
different.  There are four types of incoherent bonds, namely $(\eta_i,
\eta_j) = (-1,1), (-1,0)$, and $(1,0)$, and $(0,0)$, regardless of
order. Let us introduce the following notation: consider undirected
bonds between two particular neighboring lattice sites, $\langle
i,j\rangle$ and let $(\cdot,\cdot )$ denote the event of any of the
nine coherent or incoherent bonds so that

$$ 
(\cdot,\cdot) \in \{ (+,+), (-,-), (+,-), (-,+), (+,0), (0,+),
(-,0), (0,-), (0,0)\}.
$$

\noindent
Similarly, let $Z^{(\cdot, \cdot)}_{\mathbb T_{L}}$ denote the
partition function restricted to 
configurations where all agent spins are frozen 
in accord with the above described (chessboard) pattern
and the rest of the 
statistical mechanics
is provided by
the graffiti field against this background \cite{biskup, chessboard}. 
The following is readily obtained:
\begin{proposition}
The above described (agent--constrained) partition functions are given by

\begin{eqnarray}
\nonumber
\mathcal Z_{\mathbb T_{L}}^{(0,0)} &=&
\left[\frac{\sqrt\pi}{\sqrt\lambda}\right]^{V}, \\ 
\nonumber
\mathcal Z_{\mathbb
  T_{L}}^{(-,-)} = \mathcal Z_{\mathbb T_{L}}^{(+,+)} &=&
\left[\frac{\text{e}^{\alpha}\sqrt\pi}{\sqrt\lambda}
  \text{e}^{\frac{1}{4\lambda}[4J + K]^{2}} \right]^{V}, \\ 
\nonumber
\mathcal
Z_{\mathbb T_{L}}^{(+,-)} = \mathcal Z_{\mathbb T_{L}}^{(-,+)} &=&
\left[\frac{\text{e}^{\alpha}\sqrt\pi}{\sqrt\lambda}
  \text{e}^{\frac{1}{4\lambda}[-4J + K]^{2}} \right]^{V}, \\ 
\nonumber
\mathcal
Z_{\mathbb T_{L}}^{(0,+)} = \dots = \mathcal Z_{\mathbb T_{L}}^{(-,0)}
&=& \left[\frac{\text{e}^{\frac{1}{2}\alpha}\sqrt\pi}{\sqrt\lambda}
  \text{e}^{\frac{1}{8\lambda}K^{2}} \right]^{V}.
\end{eqnarray}
\end{proposition}

\noindent
\begin{proof}
Since the agent variables are frozen, the $g_i$ Gaussian variables are
independent and the above amount to straightforward Gaussian
integrations.
\end{proof}

\noindent
Using Lemma \ref{L:tilingBound} and the fact that the full partition
function satisfies 
$\mathcal Z_{\mathbb T_{L}} \geq \mathcal Z_{\mathbb T_{L}}^{(+,+)} $,
we can write

\begin{eqnarray}
\nonumber 
\mathbb P_{L}(0,0) &\leq&
\text{e}^{-\frac{1}{2}\alpha}\text{e}^{-\frac{1}{8\lambda}[4J +
    K]^{2}}, \\ 
\mathbb P_{L}(+,-) = \mathbb P_{L}(-,+) &\leq&
\text{e}^{-\frac{2JK}{\lambda}}, \\ 
\nonumber \mathbb P_{L}(+,0) =
\dots = \mathbb P_{L}(0,-) &\leq&
\text{e}^{-\frac{1}{4}\alpha}\text{e}^{-\frac{2J^{2} + JK +
    \frac{1}{16}K^{2}}{\lambda}}.
\end{eqnarray}

\noindent
We denote by $\varepsilon = \varepsilon (J,K,\lambda,\alpha)$ the sum
of the estimates for the probabilities provided by the
right hand sides of the preceding display.  For fixed
$\alpha$ and $K > 0$, note that as
$J\lambda^{-1/2}\to\infty$ (or, better yet, $J\lambda^{-1/2}$ and
$K\lambda^{-1/2}$ both tending to infinity) the quantity $\varepsilon$
tends to zero.  This implies the suppression of all incoherent bonds
so that the lattice must be almost fully tiled with coherent ones.
In particular, the lattice is nearly filled with agents,
which, at least locally, are mostly of the same type.  As will be demonstrated
below, this implies the existence of distinctive red and blue phases,
i.e., in the language of statistical mechanics, of a ``low temperature''
regime. We formalize this result in the next subsection.

\subsubsection{The contour argument}

We have now established all the tools we need to complete the
contour argument.  Accordingly, we now show that two 
well--separated lattice sites
must, with probability tending to one,
have identical agent spins in the limit $\varepsilon \ll 1$.  This in
turn will imply the existence of a low temperature phase.

\begin{theorem}
Consider the GI--system on $\mathbb Z^{2}$ and let
$\varepsilon(J,K,\lambda,\alpha)$ denote the quantity described in the
last paragraph of the previous subsection.  Then, if the parameters
are such that $\varepsilon$ is sufficiently small, there are at least
two distinct limiting Gibbs states characterized, respectively, by the
abundance of red agents and the abundance of blue agents. Moreover,
this property holds in any limiting shift invariant Gibbs state.
 \end{theorem}
\begin{proof}
Let us start on $\mathbb T_{L}$ with $L = 2^{k}$.  For $i,j
\in \mathbb T_{L}$ where $i$ and $j$ are well separated, let us 
consider the event 
$v_{B} :=\{\eta_{i}\neq\eta_{j}\}\cup \{\eta_{i} =
0\}$.  We will show, under the stated conditions, that uniformly in
$L$ this probability vanishes as $\varepsilon \to 0$.
As discussed previously, in order for this event to occur, the sites
$i$ and $j$ must be separated by a closed contour consisting of bonds
dual to incoherent bonds.  For $\ell = 4, 6, \dots$ let $\mathfrak
N_\ell = \mathfrak N_\ell(i-j, L)$ denote the \textit{number} of such
contours of length $\ell$ on $\mathbb T_{L}$.  Then we claim that
uniformly in $L$ and $i-j$,

\begin{equation}
\nonumber
\mathfrak N_\ell \leq 2\ell^{2}\lambda_{2}^{\ell}
\end{equation}
\noindent
where $\lambda_{2}$ (with $\lambda_{2} \approx 2.638... < 3$) is the
connectivity constant for $\mathbb Z^{2}$ \cite{connectivity}.  
A word of explanation may
be in order.  The $\lambda_{2}^{\ell}$ generously accounts for walks
of length $\ell$ in the vicinity of site $i$ and the factor of two for
walks in the vicinity of site $j$.  Finally, the factor of $\ell^{2}$
accounts for the origin of the walk.  Note this is an over-counting,
e.g., contours which wind the torus but do not necessarily
``enclose'' $i$ or $j$ are counted twice.
Using Lemma \ref{XDS} we may now write
\begin{equation}
\nonumber
\mathbb P_{L}(v_{B})
\leq   \sum_{\ell} \mathfrak N_\ell \varepsilon^{\ell}
\leq  2\sum_{\ell:\mathfrak N_\ell\neq 0}
\ell^{2}[\lambda_{2}\varepsilon]^{\ell}.
\end{equation}
The above obviously tends to zero as $\varepsilon \to 0$
demonstrating that in finite volume, the lattice is either populated
with mostly red agents \textit{or} mostly blue agents depending -- with high probability -- on what is seen at the origin.  The
implication of this result is that, for $\varepsilon$ sufficiently
small, there are at least two infinite volume Gibbs states -- which
can be realized as the limits of the appropriately conditioned
$\mathbb T_{L}$'s.  These states have one of the two mutually
exclusive characteristics: a preponderance of red agents or a
preponderance of blue agents.  The fact that the above must also hold
in any shift--invariant Gibbs state is the subject of Theorem 2.5 and
its Corollary in \cite{BbKk} with a slight extension provided by
Corollary 5.8 in \cite{C2}.
\end{proof}

\subsection{High temperature phase}
As is sometimes (e.g., historically) the case in statistical
mechanics, it can be an intricate job to establish a high temperature
phase -- a region of parameters where the limiting Gibbs measure is
unique and correlations decay rapidly.  
Typically, one calls
upon the Dobrushin uniqueness criterion \cite{D}. However for us, this
route is interdicted by the unbounded nature of the $g_i$ graffiti field.
The strategy here will be percolation based: First we establish the
so--called FKG property for all the associated Gibbs measures.  Then what
follows will be a relatively standard argument through which we show that 
the necessary and sufficient condition for uniqueness is that the
average of $\eta_{i}$ -- akin to a magnetization -- vanishes in the
state designed to optimize this quantity.  Then, finally, we will
develop a random cluster--type expansion demonstrating that under the
expected high--temperature conditions for the couplings, e.g.,
$\lambda \gg 1$, the stated condition on this magnetization is
satisfied.  In addition, high--temperature behavior should also be achieved under
the condition that agents are sparse.  This requires an alternative
percolation criterion used in conjunction with the above mentioned
expansion.  In both scenarios, the rapid decay of correlations arises as an 
automatic byproduct.

\subsubsection{FKG properties}
\label{FKGprop}
In this paragraph we will demonstrate that the \textit{FKG Lattice
  Condition} (see e.g., \cite{Li} Page 78) is satisfied by
any finite volume Gibbs measure associated with the
GI--Hamiltonian. 
Let us start by noting that we can define a natural
partial ordering on the pair of states $s_{i}$ and $s_{i}^{\prime}$
via the notation
\begin{equation}
\nonumber
s_{i}
\succeq
s_{i}^{\prime}, 
\hspace{.6 cm}
\text{if}
\hspace{.3 cm} \eta_{i} \geq \eta_{i}^{\prime}
\hspace{.4 cm}
\text{and}
\hspace{.4 cm}
g_{i} \geq g_{i}^{\prime}.
\end{equation}
Further we introduce the notation 
$\mathbf{s} \succeq \mathbf{s}^{\prime}$ to signify that the
above holds for all the $s_{i}$, $s_{i}^{\prime}$ at each
$i\in\Lambda$.  For individual spins $s_{i}$ and $s_{i}^{\prime}$, we
also denote $s_{i}\vee s_{i}^{\prime}:= (\text{max}\{ \eta_{i},
\eta_{i}^{\prime} \}, \text{max}\{ g_{i}, g_{i}^{\prime} \} )$ and
similarly for the ``minimum'' $s_{i}\wedge s_{i}^{\prime}$.  Finally,
for spin configurations $\mathbf{s}$ and $\mathbf{s}^{\prime}$, the
configurations $\mathbf{s} \vee \mathbf{s}^{\prime}$ and $\mathbf{s}
\wedge \mathbf{s}^{\prime}$ are defined as 
the sitewise maximum and minimum, respectively.
The FKG lattice condition -- conveniently stated for finite volume
measures -- is that for all $\mathbf{s}$, $\mathbf{s}^{\prime}$, the
following inequality holds:

\begin{equation}
\label{CZW}
\mathbb F_{\Lambda}(\mathbf{s}\vee \mathbf{s}^{\prime})\mathbb
F_{\Lambda}(\mathbf{s}\wedge\mathbf{s}^{\prime}) \geq\mathbb
F_{\Lambda}(\mathbf{s})\mathbb F_{\Lambda}(\mathbf{s}^{\prime}).
\end{equation}

\noindent
The well known consequence of the above is that any pair of random
variables that are both increasing with respect to the partial order
described above are positively correlated.

\begin{proposition}
\label{3point4}
The finite volume Gibbs measures associated with the GI--Hamiltonian
satisfy the FKG lattice condition.
\end{proposition}
\begin{proof}
We consider an arbitrary graph and, as will be made evident, the proof
automatically accounts for any fixed boundary conditions.  Now,
as is well known, it is sufficient to establish that the lattice
condition Eq.(\ref{CZW}) holds when differences between configurations
are exhibited only on a pair of spin--variables.  The fixed boundary
spins thus may be regarded as part of the background which is common
to all four possible agent--graffiti spin configurations in question.
Let us thus assume that the differences between two configurations
occur at sites $a$ and $b$ in the graph where certain specified
variables have been ``raised'' above a base configuration level
$\mathbf{s}$. We denote the single raise configurations by
$\mathbf{s}_{a}$ and $\mathbf{s}_{b}$ and the double raise by
$\mathbf{s}_{ab}$.  Thus, it is sufficient to show $\mathbb
F(\mathbf{s}_{ab})\mathbb F(\mathbf{s}) \geq \mathbb
F(\mathbf{s}_{a})\mathbb F(\mathbf{s}_{b})$.  All told, there are
three possibilities to consider: graffiti--graffiti, gang--graffiti and
gang--gang raises on the $a$ and $b$ sites.  For the mixed
gang-graffiti case we must also consider the $a = b$ possibility where
the gang and graffiti spins have been ``raised'' at the same site.  We
need not consider the normalization constant in any of these cases,
since it appears in identical roles on both sides of the purported
inequality; consideration of the Boltzmann factors is sufficient.  Let
us introduce, in the setting of our general graph, the interaction
\begin{equation}
\nonumber
 -\mathscr H(\mathbf{s}) = \sum_{\langle i,j \rangle}J_{i,j}
\eta_{i}g_{j} - \sum_{i}[\alpha_{i}\eta_{i}^{2} +
  \lambda_{i}g_{i}^{2}],
\end{equation}
where the first sum now extends over all edges considered to be
part of the graph and our only stipulation is that $J_{i,j} > 0$.
Also, we may formally include $i = j$ in this sum.  Let us denote the
``raised" graffiti variables via the positive increments $\delta g_a$
and $\delta g_b$ so that, in the graffiti--graffiti case, at sites $a$
and $b$ $g_{a} \to g_{a} + \delta g_{a}$ and $g_{b} \to g_{b} + \delta
g_{b}$. It is straightforward to see that
$\mathscr{H}(\mathbf{s}_{ab}) + \mathscr{H}(\mathbf{s}) =
\mathscr{H}(\mathbf{s}_{a}) + \mathscr{H}(\mathbf{s}_{b})$ and the
desired inequality holds as an identity.  Similarly for the gang--gang
case.  We can now consider the mixed case where, without loss of
generality, $g_{a} \to g_{a} + \delta g_{a}$ and $\eta_{b} \to
\eta_{b} + \delta\eta_{b}$, and for us, $\delta\eta_{b} \equiv
1$. Here
$$ -(\mathscr{H}(\mathbf{s}_{a}) - \mathscr{H}(\mathbf{s})) =
\sum_{i\neq b}J_{i,a}\eta_{i}\delta g_{a} +J_{a,b}\eta_{b}\delta g_{a} - \lambda \delta {g_{a}}^2,
$$
while
$$
-(\mathscr{H}(\mathbf{s}_{b}) - \mathscr{H}(\mathbf{s})) = \sum_{j\neq b}J_{b,j}g_{j}\delta \eta_{b}
+J_{a,b}g_{a}\delta\eta_{b} + \alpha \delta {\eta_{b}}^2.
$$
However
\begin{eqnarray}
\nonumber
-(\mathscr{H}(\mathbf{s}_{ab}) - \mathscr{H}(\mathbf{s})) 
&=& [\sum_{i\neq b}J_{i,a}\eta_{i}\delta g_{a} +J_{a,b}\eta_{b}\delta g_{a} - \lambda \delta {g_{a}}^2] +
 [\sum_{j\neq b}J_{b,j}g_{j}\delta \eta_{b}
+J_{a,b}g_{a}\delta\eta_{b} + \alpha {\eta_{b}}^2]
+J_{a,b}\delta g_{a}\delta\eta_{b} \\
\nonumber
&\geq & 2 \mathscr{H}(\mathbf{s}) - 
 \mathscr{H}(\mathbf{s}_a) - \mathscr{H}(\mathbf{s}_b).  
\end{eqnarray}
\noindent
Combining the above results we find that indeed
\begin{eqnarray}
\nonumber \mathscr{H}(\mathbf{s}_{ab}) + \mathscr{H}(\mathbf{s})
\leq \mathscr{H}(\mathbf{s}_a) + \mathscr{H}(\mathbf{s}_b).
\end{eqnarray}
\noindent
The same inequality can be easily shown in the mixed gang-graffiti
case for $a=b$, by assuming 
$g_a \to g_a + \delta g_a$ and $\eta_a \to \eta_a +
\delta \eta_a$ and by following the same steps as above. 
This completes the proof.
\end{proof}

\noindent
As an immediate consequence, we can identify boundary conditions on
$\Lambda$ which most favor the dominance of the red gang.  Indeed, it
is now seen -- as was anyway clear heuristically -- that we must make
the boundary spins ``as red as possible'' in order for a predominance
of $\Lambda$ sites to be occupied by red agents.  This amounts,
somewhat informally, to setting $g_{i} \equiv +\infty$ and $\eta_{i}
\equiv 1$ (which is anyway automatic if $K\neq 0$) all along the
boundary.  This ``specification'' which seems a bit arduous to work
with is not nearly as drastic as it sounds.  Let us start with some
notation: For $\Lambda$ a finite subset of $\mathbb Z^{2}$, let us
define $\partial \Lambda$ as those sites in $\Lambda^{c}$ with a
neighbor in $\Lambda$ and $\textsc{d} \Lambda$ as those sites in
$\Lambda$ with a neighbor in $\Lambda^{c}$.  Clearly the only
immediate consequence of the ``drastic'' boundary condition is to
force $\eta_{i} \equiv 1$ for $i\in\textsc{d}\Lambda$ and to bias, by
at most $(3J +K)g_{i}$ the \textit{a priori} Gaussian distribution of
the $g$'s.  We shall do that -- and a bit more -- on $\partial
\Lambda$ arguing that this, at most, is the result of the ``drastic''
boundary condition on $\partial(\Lambda \cup \partial \Lambda)$.
Precisely, we define the \textit{red} boundary condition on $\Lambda$
as $\eta_{i} \equiv 1$ and $g_{i}$ independently distributed as normal
random variables with variance $1/{2\lambda}$ and mean $(4J +
K)/2\lambda$ for each $i\in\partial \Lambda$.  By the established
monotonicity properties these are exactly the boundary
conditions imposed on the slightly larger lattice that will optimize
the average of $\eta_{i}$ and $g_{i}$ for any $i\in \Lambda$.

\subsubsection{A uniqueness criterion}

\noindent It is not hard to show, by monotonicity, that a limiting
\textit{red measure} exists along any thermodynamic sequence of
volumes and that the limit is independent of the sequence and
therefore translation invariant.  We shall denote this measure by
$\mu_{\textsc{R}}(\cdot)$ and by $\mathbb E_{\text{R}}(\cdot)$ the
corresponding expectations. Similarly for the blue measure
we introduce $\mu_{\textsc{B}}(\cdot)$ and $\mathbb E_{\text{B}}(\cdot)$.
We can thus state

\begin{proposition}
\label{REH}
The necessary and sufficient condition for uniqueness among the
limiting Gibbs states for the GI--system is that $\mathbb
E_{\text{R}}(\eta_{0}) = 0$, where
$\eta_0$ is the spin at the lattice origin.
\end{proposition}
\begin{proof}
For two measures $\mu_{1}$ and $\mu_{2}$ e.g., on $\{-1, 0,
1\}^{\mathbb Z^{2}}$, we use the notation $\mu_{1} \geq \mu_{2}$ to
indicate that for any random variable $X$ which is increasing 
in all coordinates, the expected values $\mathbb E_{1}(X)$, 
calculated via the $\mu_1$ measure
are always greater than the those obtained via $\mu_2$:
$$
\mathbb E_{1}(X)  \geq  \mathbb E_{2}(X).
$$
\noindent
This is known as \textit{stochastic dominance}.  Consider
$\mu_{\text{R}}(\cdot)$ which, by slight abuse of notation, we
temporarily take to be the restriction of $\mu_{R}$ to agent events.
Suppose that $\mathbb E_{\text{R}}(\eta_{0}) = 0$. Then, by
translation invariance, we have $\mathbb E_{\text{R}}(\eta_{i}) = 0$
for all $i$.  Similar considerations apply to the corresponding
$\mu_{\text{B}}(\cdot)$. It is immediately clear -- by symmetry or
stochastic dominance -- that $\mathbb P_{\text{R}}(\eta_{0} = 0) =
\mathbb P_{\text{B}}(\eta_{0} = 0)$ and thus the single site
distributions are identical.  By the corollary to the Strassen theorem
\cite{St,LP} since $\mu_{\text{R}} \geq
\mu_{\text{B}}$ \textit{and} these measures have identical single site
distributions they must be identical probability measures.  Similar
considerations apply to the full measures since the distribution of
the $g_{i}$ is determined by their conditional distributions given the
local configuration of the $\eta_i$'s.  Uniqueness is established since,
if $\mu_{\odot}(\cdot)$ denotes any other infinite volume measure
associated to the GI--Hamiltonian, we have $\mu_{\text{R}} \geq
\mu_{\odot} \geq \mu_{\text{B}}$ which implies equality in light of
$\mu_{\text{R}} = \mu_{\text{B}}$.
\end{proof}  

\subsubsection{Proof of a high--temperature phase}  
We shall develop a graphical representation for the GI--system akin to
the FK representation for the Potts model \cite{FK} that, for all
intents and purposes, is the same as the one used in \cite{BCG},
where only the case of bounded fields is explicitly analyzed.
Let us then consider the GI--Hamiltonian in finite volume with 
all notation pertaining to boundary conditions temporarily suppressed.  For
fixed $\mathbf{s}$, we may decompose the graffiti fields and agents
according to affiliation:
\begin{eqnarray}
\nonumber
g_{i}  &=&  q_{i}\vartheta_{i}; \hspace{.25 cm} \vartheta_{i} = \pm 1, \hspace{.15 cm} q_{i} = |g_{i}|, \\
\nonumber
\eta_{i} &=& r_{i}\sigma_{i}; \hspace{.25 cm} \sigma_{i} = \pm 1, \hspace{.15 cm} r_{i} = |\eta_{i}|,
\end{eqnarray}
where the $\sigma$'s and $\vartheta$'s have the definitive character 
of \textit{Ising} variables. We can now write
$$
\text{e}^{J_{i,j}g_{i}\eta_{j}} = 
\text{e}^{-J_{i,j}q_{i}r_{j}}(R_{i,j}\delta_{\vartheta_{i},\sigma_{j}} + 1),
$$ where $R_{i,j} = R(J_{i,j},q_{i},r_{j}) :=
\text{e}^{2J_{i,j}q_{i}r_{j}} -1$.  In our case, we have $J_{i,j} = J$
if $i$ and $j$ are neighboring pairs and $J_{i,i} = K$; which we will
not yet distinguish notationally and consider a general $J_{i,j}$
label. Thus

$$
\text{e}^{-\mathscr H(\mathbf{s})} =
\prod_{(i,j)}\text{e}^{-J_{i,j}q_{i}r_{j}}(R_{i,j}\delta_{\vartheta_{i},\sigma_{j}}+1).
$$

\noindent
Opening the product, we select one term for each ``edge'': If the
$R_{i,j}$ term is selected, we declare the edge to be
\textit{occupied}, otherwise it is \textit{vacant}.  It is noted here
that the edges should be interpreted as \textit{directed}:
all edges appear twice and we must regard $\langle
i,j\rangle$ as distinctive from $\langle j,i\rangle$; moreover, for $K
\neq 0$, the above is understood to include $i = j$.
The configurations of occupied edges will, generically, be denoted by
$\omega$.  Summing over the Ising variables, we acquire the weights
$$
W(\omega)  =  2^{C(\omega)}\sum_{\mathbf{q},\mathbf{r}}\prod_{(i,j)\in\omega}
R_{i,j}(J_{i,j}q_{i},r_{j}),
$$
\noindent
where, as before the summation notation also indicates integration
over the continuous variables.  In the above, $C(\omega)$ denotes the
number of connected components of $\omega$; here connectivity
deduced according to the \textit{directed} nature of the edges or via
a double covering of the lattice.
Normalizing these weights by the partition function we obtain a
probability measure on the bond configurations $\omega$.  As will be
made explicit below, this probability measure on bond configurations
is well defined in finite volume.  Let us denote the
probability measure on the bond configurations $\omega$
by $\mathbb P_{\Lambda}^{\odot}(\cdot)$,  where the $\odot$
now denotes boundary conditions accounted for in a routine fashion.
Then, for each $\omega$ consisting of appropriate edges, $\mathbb
P_{\Lambda}^{\odot}(\omega) \in (0,1)$.  We shall not discuss the
problem of infinite volume limits which would take us too far astray
but be content with statements that are uniform in volume.  With
regards to the latter, and of crucial importance for our purposes is
the connection back to the spin--measure inherent in this
representation. For the Potts models, this was first elucidated in \cite{ACCN} with the
complete picture emerging in \cite{ES}. In particular, for any site,
the contribution to the magnetization vanishes if the site
belongs to a cluster that is isolated from the boundary.
The principal objective for this representation is the following claim:

\begin{proposition}
\label{PIY}
Let $\Lambda \subset \mathbb Z^{2}$ be a finite connected set and
consider the above described representation in $\Lambda$ with boundary
condition $\odot$ on $\partial \Lambda$.  Let $\langle a,b\rangle$
denote an edge with $a\neq b$ and both $a$ and $b$ not belonging to
$\partial \Lambda$.  Let $\mathbf{e}_{ab}$ denote the event that this
edge is occupied and let $\omega$ denote a configuration on the
compliment of $\langle a,b\rangle$.  Then, for fixed $\alpha$ and $K$,
there is an $\varepsilon(J,\lambda)$ with $\varepsilon \to 0$ as
$J^{2}/\lambda\to 0$ such that uniformly in $\Lambda$, $\omega$ and
$\odot$ -- as well as $K$ and $\alpha$,
$$
\mathbb P_{\Lambda}^{\odot}(\mathbf{e}_{ab}) < \frac{\varepsilon}{1 + \varepsilon}.
$$
\end{proposition}
\begin{proof}  
Let $W_\Lambda^{\odot}(\cdot)$ denote the configurational weights with
associated boundary conditions as described above.  Then it is seen
that
$$ 
\frac{\mathbb P_{\Lambda}^{\odot}(\omega \vee \mathbf{e}_{ab})}{1 -
  \mathbb P_{\Lambda}^{\odot}(\omega \vee \mathbf{e}_{ab})} =
\frac{W_\Lambda^{\odot}(\omega \vee
  \mathbf{e}_{ab})}{W_\Lambda^{\odot}(\omega)}.
$$ 
Our goal is to estimate the right hand side of the above
which thereby generates the quantity
$\varepsilon$ featured in the statement of this proposition.  Noting
the positivity and product structure of the numerator and
denominator, we may regard the object on the right as the expectation
with respect to a weighted measure of the quantity $R_{ab}$ and we
shall denote this by $\mathbb E_{\omega}(R_{a,b})$.  The latter will
be estimated via conditional expectation: Let $Q_{\hat{q}_{a}}$ denote
a specification of the $q$--fields and agent occupation variables
except for $q_{a}$ and let
\begin{eqnarray}
\nonumber
\varepsilon := \sup_{\omega, Q_{\hat{q}_{a}}}\mathbb E_{\omega}(R_{ab}\mid Q_{\hat{q}_{a}}).
\end{eqnarray}
Obviously, $\varepsilon \geq \sup_{\omega}\mathbb E_{\omega}(R_{ab})$.
As for the complimentary fields, there is not a great deal of
dependence: In particular, all that is needed is that $r_{c} = 1$ for
all $c$ such that $\langle a, c\rangle \in \omega$.  Concerning the
optimizing $\omega$, non--local considerations dictate simply, that
$\omega$ be such that $\langle a,b\rangle$ does not \textit{reduce}
the number of components.  Locally, as can be explicitly checked, or
derived from monotonicity principals, the optimal scenario is when all
bonds emanating from $a$ are present in the configuration.  Thus we
have
\begin{eqnarray}
\nonumber
\varepsilon = 
\frac{\int\text{e}^{-(4J+ K)q}R^{4}(J)R(K)\text{e}^{-\lambda q^{2}}dq}
{\int\text{e}^{-(3J+K)q}R^{3}(J)R(K)\text{e}^{-\lambda q^{2}}dq}
=
2
\frac{\int\text{e}^{-\lambda q^{2}} \sinh^{4}(Jq)\sinh (Kq) dq}
{\int\text{e}^{-\lambda q^{2}} \sinh^{3}(Jq) \sinh (Kq)dq},
\end{eqnarray}
where in the first line $R(J) := R(J, q, 1)$.  We claim that the
final ratio is bounded by $J/\lambda^{1/2}$ multiplied by a constant
that may be proportional to the ratio $K/\lambda^{1/2}$.  Indeed let
us substitute $\omega = \lambda^{1/2}q$ and $\kappa :=
K/\lambda^{1/2}$.  The above quantity can thus be rewritten as

\begin{eqnarray}
\nonumber
\varepsilon  =  2\frac
{{\int\text{e}^{-\omega^{2}}\sinh \kappa\omega(\sinh \omega J/\lambda)^{4}d\omega}}
{\int\text{e}^{-\omega^{2}}\sinh \kappa\omega(\sinh \omega J/\lambda)^{3}d\omega}.
\end{eqnarray}

\noindent
Our claim is obvious if $\kappa\to 0$ but we may wish to consider
cases where $\kappa$ stays bounded away from zero.  In general, the
integrands are not dominated by large $\omega$ and we may expand the
factors $\sinh \omega J/\lambda$ with the result
\begin{eqnarray}
\nonumber
\varepsilon \to \frac {2 J} {\lambda^{1/2}}
\frac{\int\text{e}^{-\omega^{2}}\sinh \kappa\omega\cdot
  \omega^{4}d\omega}{\int\text{e}^{-\omega^{2}}\sinh \kappa\omega\cdot
  \omega^{3}d\omega}.
\end{eqnarray}
We finally claim is that the right side is bounded by a linear
function of $\kappa$:
\begin{eqnarray}
\nonumber
\frac{1}{1 + \kappa} \frac{\int\text{e}^{-\omega^{2}}\sinh
  \kappa\omega\cdot \omega^{4}d\omega}{\int\text{e}^{-\omega^{2}}\sinh
  \kappa\omega\cdot \omega^{3}d\omega} < B,
\end{eqnarray}
for some $B < \infty$. This is indeed true as $\kappa \to 0$. We only
need to show that the inequality holds in the case 
$\kappa\to \infty$.  But here the factor
$\text{e}^{-\omega^{2}}\sinh \kappa\omega$ is, essentially, a Gaussian
in the variable $\omega - \kappa$ and the desired result follows.
\end{proof}
\begin{theorem}
\label{PKS}
Consider the GI--system and let $\varepsilon$ denote the quantity
described in Proposition \ref{PIY}.  Then for $\varepsilon <
\varepsilon_{0}$, given by
$$
\varepsilon_{0} + \frac{1}{2}\varepsilon_{0}^{2} = \frac{1}{2}
$$
there is a unique limiting Gibbs state featuring rapid decay of correlations.
\end{theorem}
\begin{proof}
Using the result of Proposition \ref{PIY}, we shall compare the
described graphical representation with independent bond percolation
on $\mathbb Z^{2}$.  We start with a well known -- and readily
derivable -- result: Let $Y_{1}, \dots Y_{N}$ denote an array of
Bernoulli random variables with collective behavior described by the
measure $\mu_{\mathbf{Y}}$ and let $\mathbf{Y}_{\hat{Y}_{j}}$ 
denote a configuration on the compliment of $Y_{j}$. Let us now introduce
\begin{eqnarray}
\nonumber
 p_{j} = \max_{\hspace{3 pt}\mathbf{Y}_{\hat{Y}_{j}}} \mathbb
P_{\mathbf{Y}}(Y_{j} = 1\mid \mathbf{Y}_{\hat{Y}_{j}})
\end{eqnarray}
\noindent
to denote the maximal conditional probability of
observing $\{Y_{j} = 1\}$.  Finally, let $X_{1}, \dots X_{N}$ denote a
collection of independent Bernoulli random variables with parameters
$p_{1}, \dots , p_{N}$.  Then, denoting the independent measure by
$\mu_{\mathbf{X}}$, we have
\begin{equation}
\nonumber
\mu_{\mathbf{X}} \geq \mu_{\mathbf{Y}}.
\end{equation}
\noindent
Thus we may bound the probabilities of increasing events in the
graphical representation by the corresponding probabilities from
independent percolation on $\mathbb Z^{2}$ with bond occupation
probabilities determined by the $\varepsilon$ from Proposition
\ref{PIY}.  However, we must note that the relevant percolation problem
has multiple types of edges.
The $\varepsilon_{0}$ in the statement of
this proposition bounds the probability of the event
$\mathbf{e}_{ab}\cup\mathbf{e}_{ba}$ by $\frac{1}{2}$.
If the
clusters of the featured representation fail to percolate, then, as
$\Lambda \nearrow \mathbb Z^{2}$, the origin is disconnected from the
boundary with a probability tending to one.  As discussed just prior
to Proposition \ref{PIY}, this implies $\mathbb E_{\text{R}}(\eta_{0})
= 0$ and by Proposition \ref{REH}, uniqueness is established.
Under the condition $\varepsilon < \varepsilon_{0}$, exponential decay
of correlations can also be established.  We will be content with the
decay of the two point function.  The problem of general correlations
under these conditions has been treated elsewhere
\cite{MC1,MC2}. In particular, for $i,j\in\mathbb Z^{2}$,
$\mathbb E(\eta_{i}\eta_{j})$ in the unique infinite volume measure is
bounded, in finite volume approximations by the probability that $i$
and $j$ reside in the same cluster.  For $\varepsilon <
\varepsilon_{0}$, this decays exponentially in $|i-j|$ uniformly in
$\Lambda$ for $|\Lambda|$ sufficiently large.
\end{proof}

\noindent
We now turn our attention to an alternative criterion for high
temperature behavior which may also be of relevance in a sociological
context: Sparsity of agents.  Mathematically, this pertains to the
situation where $\alpha$ is large and negative ($-\alpha \gg 1$) which
\textit{a priori} suppresses the fraction of agent occupied sites.
Our arguments will initially be based on more primitive notions of
percolation and, following the methods of \cite{CNPR} (see also
\cite{C,CMW}) could, perhaps, be completed along these lines.
However, it turns out to be far simpler to appeal to the graphical
representation just employed for the final stage of the argument.  We
start with the relevant notion of percolation and connection. In the
context of site percolation on $\mathbb Z^{2}$, we may define various
notions of connectivity \cite{Grimmett}.  
Here we define $\diamond$--connectivity to
indicate connection between sites that are no more than two lattice
sites away. This is not to be confused with $\ast$--connectivity which
does not consider a pair of sites to be connected if they are
separated by two units in the vertical or horizontal direction.  We
denote by $p_{c}^{\diamond}$ the threshold for $\diamond$--percolation
on $\mathbb Z^{2}$.  Standard arguments dating to the beginning of the
subject show that $p_{c}^{\diamond} \in (0,1)$; in particular,
$p_{c}^{\diamond}$ is less than the threshold for ordinary, or even
$\ast$--connected, percolation and mean--field type bounds readily
demonstrate that $p_{c}^{\diamond} > \frac{1}{12}$.

The next proposition concerns the relative abundance of, e.g., red sites under the condition $-\alpha \gg 1$ with the other parameters fixed.  

\begin{proposition}
\label{JUO}
Consider the GI--system with parameters $\lambda$, $K$ and $J$ fixed.
Then there is a $\delta_{\alpha} = \delta_{\alpha}(J,K,\lambda)$ with
$\delta_{\alpha} \to 0$ as $\alpha\to-\infty$ such that uniformly in
volume and boundary conditions, for any site $i$ that is away from the
boundary
$$
\mathbb P^{\odot}_{\Lambda}(\eta_{i} = 1) < \delta_{\alpha}.
$$
\end{proposition}
\begin{proof}
Here we employ the preliminary (red $\succeq$ blue) FKG properties
that were established earlier, in \ref{FKGprop}.  
We start with a $\gamma > 0$ (and
somewhat ``large'') and, for $j\in \Lambda$ not too near the boundary,
we consider $\mathbb P^{\odot}_{\Lambda}(g_{j} > \gamma)$.  By the FKG
property, this probability is less than the corresponding conditional
one given that $\eta_{j} = 1$ and that $\eta_{k} = 1$ for all $k$ that are
neighbors of $j$.  This conditional probability is given by a
definitive expression:

\begin{equation}
\nonumber
\mathbb P^{\odot}_{\Lambda}(g_{j} > \gamma) \leq 
\frac{
\int_{g > \gamma}\text{e}^{+(4J + K)g}\text{e}^{-\lambda g^{2}}dg}
{\int_{g}\text{e}^{+(4J + K)g}\text{e}^{-\lambda g^{2}}dg}
:= \delta_{\gamma}
\end{equation}

\noindent
The above can be expressed directly via the error function but in
any case, as is not hard to show,
\begin{equation}
\nonumber
\delta_{\gamma} \leq \frac{1}{2}\text{e}^{-\lambda[ \gamma -\frac{4J + K}{2\lambda}]^{2}},
\end{equation}
\noindent
as long as $\gamma \geq (4J + K)/2\lambda$, which also quantifies
how large $\gamma$ must be. Provided $i$ is a few spaces away from the
boundary, we note that $1 - 5\delta_{\gamma}$ is a valid
estimate of the probability that 
both $g_{i}$ and the $g$--values at the neighbors of $i$
do not exceed $\gamma$.  Let us denote this (good) non--high field
event by $G_{i}$.  Then we may write
\begin{align}
\mathbb P_{\Lambda}^{\odot}(\eta_{i} = 1) & = \mathbb
  P_{\Lambda}^{\odot}(G_{i})\mathbb P_{\Lambda}^{\odot}(\eta_{i} =
  1\mid G_{i}) + \mathbb P_{\Lambda}^{\odot}(G_{i}^{c})\mathbb
  P_{\Lambda}^{\odot}(\eta_{i} = 1\mid G_{i}^{c}) \notag \\ 
&\leq
  5\delta_{\gamma} + \frac{\text{e}^{(4J +
      K)\gamma}\text{e}^{\alpha}} {1 + \text{e}^{(4J +
      K)\gamma}\text{e}^{\alpha} + \text{e}^{-(4J + K)\gamma}
    \text{e}^{\alpha}} := \delta_{\alpha},
\end{align}
where, in various stages we have employed worst case scenarios.
Clearly, for fixed $(J, K, \lambda)$ we may choose $\gamma$ large so
that $\delta_{\gamma}$ is small, and $\alpha$ negative and
large in magnitude, so that $\delta_{\alpha}$ is small.
\end{proof}

\begin{theorem}
Consider the GI--System and suppose that $-\alpha$ is large enough so
that $\delta_{\alpha} < p_{c}^{\diamond}$ as described just prior to
the statement of Proposition \ref{JUO}.  Then there is a unique
limiting Gibbs state featuring rapid decay of correlations.
\end{theorem}

\begin{proof}
By the dominance principle stated at the beginning of the proof of
Proposition \ref{PKS} if $\delta_{\alpha} < p_{c}^{\diamond}$, the red
agents fail to exhibit $\diamond$--percolation regardless of boundary
conditions. Now consider, in the context of the bond--representation, the event
that the origin is connected to $\partial \Lambda$ in the red boundary
conditions, which represents the sole non--vanishing contribution to
$\mathbb E^{R}_{\Lambda}(\eta_{0} = + 1)$.  The bonds of any path
connecting the origin to $\partial \Lambda$ within this cluster may be
envisioned as alternating connections between agents and fields; the
connection to the red boundary ensures that both types of entities
take on the red color.  In particular, all the agents in the cluster
are red so that these agents must (at least) form a
$\diamond$--connected cluster.  Hence, in finite volume, we may bound

\begin{equation}
\nonumber
\mathbb E_{\Lambda}^{\text{R}}(\eta_{0} = +1) \leq \mathbb
P_{\Lambda}^{\text{R}}(0 \underset{\diamond, \text{R}}{\leadsto}
\partial \Lambda)
\end{equation}

\noindent
where $\{0 \underset{\diamond,\text{R}}{\leadsto} \partial \Lambda\}$
is the event of a red $\diamond$--connection between the origin and
the boundary.

When the red agent occupation probabilities are dominated by
independent sites with parameter $\delta_{\alpha} < p_{c}^{\diamond}$,
such probabilities decay exponentially.  Evidently, in the limiting
state, the ``magnetization'' vanishes which by Proposition \ref{REH}
implies a unique state.  Similarly, exponential decay of correlations
is implied by exponential decay of $\diamond$--connectivities.
\end{proof}

\section{The Mean Field Rendition} \label{S:MFHamiltonian}
\noindent
In the previous section, we showed that a phase transition between
well--mixed and clustering configurations exists for the general
Hamiltonian in Eq.\,\eqref{E:Hamiltonian}.  However, finding the exact
or even approximate values of the $J,K,\alpha,\lambda$ parameters for
which the well--mixed to clustering transition occurs is in general a
difficult task.  Moreover, the nature of the transition is not
elucidated by the techniques of the preceding section.  On 
the basis of informal simulations described in the Appendix 
and certain other
considerations it appears that the transition may be discontinuous or
second order depending on where the phase boundary is crossed.  This
cannot be proved in the context of the present model.  We thus
introduce a \textit{mean--field Hamiltonian}, where instead of
nearest--neighbor interactions we consider an all--to--all
(interaction) coupling that is rescaled by the number of sites.
Models of this sort are often referred to as \textit{complete--graph}
systems. The mean--field Hamiltonian allows us to define, in the
thermodynamic limit, a simple mean field free energy per particle.
This free energy can be subjected to exact mathematical analysis
which provides a quantification of the phase transition.  In
particular, we have found that 
the phase boundary between the diffuse states and the
gang--symmetry broken phase can indeed be of either type.

Let us thus consider a lattice of $N$ sites -- where the detailed
geometry is no longer of relevance.  At each site $i$, there is the
same $s_i = (\eta_i, g_i)$ featured in the previous section.  However
now, the Hamiltonian reads
\begin{equation}
\label{MFH}
 - \mathscr{H}^{\textsc{MF}}_{N}(\mathbf{s}) = \frac{1}{N} \sum_{i,j} J \eta_i g_j 
 + \sum_{i} (\alpha
 \eta_i^2- \lambda g_i^2).
 \end{equation}
It is observed that the couplings $J$ and $K$ need no longer be
distinguished.  Indeed, for large $N$, the $g_{i}\eta_{i}$
interaction, and any other \textit{particular} interaction is not of
pertinence.  We now introduce the relevant collective quantities, $n$
$G$, and $b$, obtained via $\eta_i$ and $g_i$, which will allow for a
more convenient analysis.  In particular, if $N^+$ and $N^-$ designate
the number of red and blue lattice agents, respectively, we define
$$
b := \frac{N^+ + N^{-}}{N}
\hspace{1 cm}
\text{and}
\hspace{1 cm}
n :=
\frac{N^{+} - N^-}{N},
$$ as the fraction of the lattice covered by agents of any type and
the excess -- positive or negative -- of this fraction that is of the
red type.  Moreover, we introduce $G = \frac{1}{N} \sum_{i} g_i$ 
to be the average graffiti imbalance.  In this
context, $n$ and $G$ are akin to magnetizations in a standard
one--component spin model, with $n$ corresponding to magnetization in
the agent variables and $G$ in the graffiti field.  For occasional use, we
also define $n^{\pm} = {N^{\pm}}/{N} \leq 1$.  We remark that in these
definitions there is an implicit $N$ dependence which is notationally
suppressed.

\subsection{The partition function}
\noindent
In the forthcoming, we will evaluate, asymptotically, the mean field
partition function $\mathcal{Z}^{\text{MF}}$ defined in accord with the
previous section as the partition sum $
\mathcal{Z}^{\text{MF}}_{N} = \sum_{\mathbf{s}}
e^{-\mathscr{H}_{N}^{\text{MF}}(\mathbf{s})}.  $ Here, for reasons which
will soon become clear, we will treat the graffiti field variables 
slightly differently.  We define
$$
d\mu_{g_{i}} := \sqrt{\frac{\lambda}{\pi}}\text{e}^{-\lambda g_{i}^{2}},
$$
as the normalized Gaussian measure for the individual field variables.  
Letting $\mathbf{g}$ denote the array of these random variables
we may write
\begin{equation}
\nonumber
\mathcal{Z}^{\text{MF}}_{N} :=  \mathbb E_{\mathbf{g}}
\left(
\sum_{\mathbf{\eta}}\text{e}^{J\sum_{i,j}g_i\eta_j +\alpha\sum_{i}\eta_{i}^{2}}
\right),
\end{equation}
\noindent 
where $\mathbb E_{\mathbf{g}}(\cdot)$ denotes expectation with
respect to the free (independent) ensemble of Gaussian random
variables and $\sum_{\mathbf{\eta}}$ denotes the rest of the partition
sum i.e., over the agent configurations.  It is acknowledged that this
differs from the prior definitions by a multiplicative factor of
$[\lambda/\pi]^{N/2}$ which, of course, is inconsequential.

It is at this point, with the current formulation, that the
advantage of the all--to--all coupling is manifest: For any
$\mathbf{s}$ (and any $N$) the quantity in the exponent depends only
on $n$, $b$ and $G$: $\mathcal{Z}^{\text{MF}}_{N} = \mathbb E_{\mathbf
  g}(\sum_{\mathbf{\eta}}\text{e}^{N(JnG + \alpha b)})$.  Concerning
the agent configurations, to perform the summation, we must multiply
the integrand by the number of ways of arranging $N^{+}$ red sites and
$N^{-}$ blue sites among $N$ possible positions.  We denote this
object by $W_{N}(b,n)$ which is given, explicitly, by the trinomial
factor
\begin{equation}
\nonumber
W_{N}(b,n) = \binom{N}{N^{+},N^{-}} =
\binom{N}{\frac{1}{2}N(b+n),\frac{1}{2}N(b-n)}.
\end{equation}
\noindent
As for the graffiti field configurations, it is noted that since $G$ is
proportional to a sum of Gaussian random variables, it is itself a
Gaussian.  Indeed the mean of $NG$ is zero and the variance is
$N[2\lambda]^{-1}$.  Thus the expectation over $\mathbf{g}$ can be
replaced with the expectation over $NG$ leading to
\begin{equation}
\nonumber
\mathcal{Z}^{\text{MF}}_{N} =\sum_{n,b}\mathbb
E_{NG}[W_{N}(b,n)\text{e}^{N[JbG + \alpha b] }]
\propto
\sum_{n,b}W_{N}(n,b)\text{e}^{N[JnG + \alpha b - \lambda G^{2}]} dG
\end{equation}
with the constant of proportionality independent of $N$.
Now, on the basis of the Stirling approximation, 
\begin{equation}
\nonumber
W_{N}(b, n)
\approx 
\left[\left(\frac{b+n}{2} \right)^{\frac{b+n}{2}}
\left(\frac{b-n}{2}\right)^{\frac{b-n}{2}} (1- b)^{1-b} \right]^{-N}.
\end{equation}
Thus, modulo lower order terms, we have
$\mathcal{Z}^{\text{MF}}_{N} \approx\sum_{n,b,G}\text{e}^{-N\Phi(b,n,G)}$
where $\Phi$, the free energy function, is given by 
\begin{equation}
\label{OIQ}
e^{- \Phi(b,n,G)} := e^{(JnG + \alpha b -\lambda G^2)} \left[ 
\left(\frac{b+n}{2}\right)^{\frac{b+n}{2}} \left(\frac{b-n}{2} \right)
^{\frac{b-n}{2}} (1- b)^{1-b} \right]^{-1}.
\end{equation}

\noindent
In accordance with standard asymptotic analysis
\begin{equation}
\nonumber
\lim_{N\to\infty}-\frac{1}{N}\log \mathcal{Z}^{\text{MF}}_{N} =
\min_{b,n,G}\Phi (b,n,G) := F_{\text{MF}}
\end{equation}
\noindent
where $F_{\text{MF}} = F_{\text{MF}}(J,\alpha,\lambda )$ is the
(actual) limiting free energy per site.  While various aspects of the
above scenario for all--to--all coupling models have been long known
and certain cases explicitly proven \cite{E}, there is a general
theorem to this effect that is sufficient for our purposes,
presented in Section 5 of \cite{BC}.
Thus the efforts of a mean--field analysis may be summarized as
follows: we are to minimize $\Phi(b,n,G)$ and the values of $b$, $n$ and $G$
at the minima -- as a function of the couplings -- will determine the
various \textit{phases} of the system.  Even in this simplified
context, as will be seen, the phase transitions can be dramatic.

\subsection{The mean--field equations}
The free energy function is obviously well behaved except at the
extreme values of the variables.  In particular, we would like to
assume that $0 < b < 1$ and $-b < n < +b$ where the strict
inequalities imply that the function is smooth.  Now a direct
calculation of the asymptotics makes it clear that no minimum could
possibly occur near the $b=0,1$ and $b= \pm n$ boundaries. Thus, we can
confine attention to the interior of the above $b$ and $n$ intervals and
proceed by differentiation of $\Phi(b,n,G)$ as defined in Eq.(\ref{OIQ}).
Thus we arrive at the \textit{mean--field equations}:
\begin{align}
-\frac{\partial \Phi}{\partial G} &= J n - 2 \lambda G = 0,
\label{VSR} \\ 
-\frac{\partial
\Phi}{\partial b} &= \alpha + \log(1-b) - \frac{1}{2} \log \left(
\frac{b^2 - n^2}{4} \right) = 0,
\label{VSQ}\\ 
-\frac{\partial \Phi}{\partial n} &= J G
- \frac{1}{2} \log \left( \frac{b+n}{b-n} \right) = 0,
\label{VDP}
\end{align}
Free energy minimization only occurs for values of $(n,b,G)$ that
satisfy the above system.  However, other stationary points for $\Phi$
can -- and, e.g., in the case of discontinuous transitions generically will
-- occur so we must proceed with some caution.  It is noted that
Eq.(\ref{VSR}) allows us to eliminate $G$ altogether.  Defining $\mu =
\frac{J^2}{2\lambda}$, we rewrite Eqs.\,\eqref{VSQ} and \eqref{VDP} as
\begin{align}
4e^{2 \alpha} &= \frac{b^2-n^2}{(1-b)^2} \label{E:mf_b2},\\
\mu n &= \frac{1}{2}\log \left( \frac{b+n}{b-n} \right) \label{E:mf_n2}.
\end{align}
The analysis of this system, along with the minimization it is
supposed to imply will constitute the bulk of the remainder of this
work.  Foremost, it is noted that the presentation in
Eqs.(\ref{E:mf_b2}) and (\ref{E:mf_n2}) are, for all intents and
purposes, the same as would have been obtained from the mean--field version
of the so--called BEG model \cite{BEG}.  As such, some aspects of the
current problem have been treated in \cite{EOT}.  However, the
specifics in \cite{EOT} are not readily translated into that of the
current work and, moreover, our conclusions are achieved by
straightforward methods of analysis.

Our investigation will proceed as follows: It is evident from physical
considerations, and the subject of an elementary mathematical theorem
proved at the end of this subsection, that as the parameters sweep
through their allowed values, a phase transition occurs from the
circumstances where $\Phi$ is minimized by $n = 0$ to those where
$n\neq 0$ is required.  First, we will follow the consequences of the
\textit{assumption} that this happens continuously: i.e., that the
minimizing $n$ goes to zero continuously through small values.  In the
leading order, this provides a purported phase boundary which we
denote by the LSP--curve.
Considerations of higher order terms in the vicinity of the LSP--curve
yield that for certain portions of the curve, the stipulation is
self--consistent and for the rest, it is not.  Detailed analysis will
show that the former is completely consistent.  In particular these
calculations correspond to the true minima of the free energy
function.  By contrast, the latter (non--self--consistent) portion is
a consequence of a discontinuous transition which has ``already''
occurred at prior values of the parameters.  In particular, the
perturbative analysis is highlighting a local extremum and not the
true minimum.

We conclude this subsection with the derivation of the LSP--curve --
as well as the introduction of notation that will be used throughout
the reminder of the analysis.  Assuming $n = 0$, Eq.(\ref{E:mf_n2}) is
trivially satisfied and Eq.(\ref{E:mf_b2}) defines the ``ambient''
value of $b$ which we denote by $b_{R}$:
\begin{equation}
\nonumber
b_R := \frac{2 e^\alpha}{1 + 2 e^{\alpha}}.
\end{equation}
\noindent
Note that $(b = b_{R}, n = G = 0)$ is \textit{always} a solution to
the mean--field system. For simplicity we consider $b_R$ and $\mu$ as
the relevant parameters for our system for the remainder of this
paper.  Let us now consider slight perturbations of $b$ about $b_R$
and of $n$ about zero.  We thus write $b = b_R(1 + \Delta)$ with
$\Delta \ll b_R$ and $|n| > 0$ with $n \ll 1$ and obtain the following
approximations by expanding Eq.(\ref{E:mf_b2}) to lowest order
\begin{equation}
\label{E:relNandDelta}
n^2 \approx 2\Delta \frac{b_R^{2}}{1- b_R},
\end{equation}
while Eq.(\ref{E:mf_n2}), written to a higher approximation than will
be immediately necessary, gives us
\begin{equation}
\label{KIT}
\mu n \approx \frac{n}{b_R} - \frac{n\Delta}{{b_R}} + \frac{n^3}{3{b_R}^3}.
\end{equation}
We pause to observe that Eq.(\ref{KIT}) and, in general,
Eq.(\ref{E:mf_n2}), have the symmetry property that with all other
quantities fixed, if $n$ is a solution then so is $-n$.  Thus, we
might as well assume that $n \geq 0$.  Indeed, we shall adhere to this
convention throughout.
Assuming now that $\mu$ is variable while $b_{R}$ is fixed, the
$n\to 0$ limit of Eq.(\ref{KIT}) and Eq.(\ref{E:relNandDelta}) yields
the tentative phase boundary
\begin{equation}
\label{RDA}
\mu_{S}(b_{R}) = \frac{1}{b_{R}}.
\end{equation}
This \textit{defines} the LSP--curve; the correspondingly tentative
conclusion is that $n > 0$ and $b > b_{R}$ occurs for $\mu > \mu_{S}$
while for $\mu \leq \mu_{S}$, $n = b - b_{R} = 0$.  However, the
viability of these tentative conclusions depends, in a definitive
fashion, on the value of $b_{R}$.  In particular an analysis of the
higher order terms in Eq.(\ref{KIT}) testifies that this picture
cannot possibly be correct for $b_{R} < \frac{1}{3}$; this is the subject of
our next subsection.  However, a more difficult analysis shows that
this picture is indeed correct for $b_{R} \geq \frac{1}{3}$ which is
the subject matter of the final subsection.  First, we must attend to
some necessary details.

\subsubsection{Preliminary analysis}
In this subsection we will establish some basic properties of the
model such as the existence of high-- and low--temperature phases
along with various monotonicity properties.  In particular we show
that at fixed $b_{R}$, the quantity $n$, assumed to be non--negative,
is non--decreasing with $\mu$ and strictly increasing whenever it is
non--zero.  For the benefit of our physics readership, such
contentions might typically be \textit{assumed} and consequently, the
entire subsection could be skipped on a preliminary reading.  However,
it is remarked that in the normal (physics) course of events, such
questions are most often settled by direct perturbative calculation.
Even for continuous transitions, on some occasions, additional
justification is actually required.
Sometimes, as in the present work, when the transition is
discontinuous, the relevant calculations simply cannot be done
analytically and then, indeed, one must rely more heavily on
abstract methods.

In what follows, we shall work with the free energy function given by
Eq.(\ref{OIQ}) with $G$ eliminated in favor of $n$ according to
Eq.(\ref{VSR}) and working with the parameters $\mu$ and $b_{R}$.  
For simplicity, this will be denoted by $\Phi_{b_{R},\mu}(b,n)$ but
with subscripts omitted unless absolutely necessary.  Thus $\Phi(b,n)$
is now notation for the function
\begin{align}
\Phi_{b_{R},\mu}(b,n) :=
-\frac{1}{2}\mu n^{2} -\alpha(b_{R}) b&
\notag
\\
+\left(
\frac{b+n}{2} 
\right)
\log& \left(
\frac{b+n}{2}
\right)
+\left(
\frac{b-n}{2} 
\right)
\log \left(
\frac{b-n}{2} 
\right)
+(1-b)\log(1-b).
\end{align}
It is clear that the minimum of $\Phi(b,n)$ corresponds to the minimum
of the original three variable free energy function $\Phi(n,b,G)$.  
In the following, will use the notation $n(\mu)$ (with $n(\mu)
\geq 0$) as though this defines an unambiguous function.  Of course in
the case of phase coexistence, this will not be true.  In general,
then, $n(\mu)$ will stand for a representative from the set of
minimizers at parameter value $\mu$ and all of the results in this
subsection hold.
We start with some elementary properties of the phase diagram 
generated by the corresponding minimization problem.
\begin{proposition}
\label{YTZ}
Consider $\Phi(b,n)$ with $b_{R}$ fixed and $\mu$ ranging in
$[0,\infty)$.  Then for all $\mu$ sufficiently large, $\Phi(b,n)$ is
  minimized by a non--zero $n$ and for all $\mu$ sufficiently small,
  $\Phi$ is minimized by $(b_{R},0)$.
\end{proposition}
\begin{proof}
We begin with the assertion, gleaned from Eq.(\ref{VSQ}), that along
the curves $n = 0$, $\Phi$ is minimized by $b = b_{R}$.  Thus we may
pick any fixed, nontrivial $n_{0}$, with $0 < n_{0} < b_{R}$, and it
is sufficient to establish that $\Phi(n_{0}, b_{R}) < \Phi(0, b_{R})$
once $\mu$ is sufficiently large.  However, the desired inequality is
manifest for large $\mu$ since the only $\mu$ dependence in $\Phi$ is
in the term $-\frac{1}{2}\mu n_{0}^{2}$ which is, eventually, in
excess of the differences between the $\mu$ independent term and
$\Phi(0, b_{R})$.
The second statement is proved as follows: since $b = 0$ -- which
necessarily implies $n = 0$ -- does not minimize the free
energy function, we may use the variable $\theta := n/b$ so that
Eq.(\ref{E:mf_b2}) now reads
\begin{equation}
\nonumber
b\mu \theta = 
\frac{1}{2}\log
\left(\frac{1 + \theta}{1 - \theta}\right ).
\end{equation}
As is well known from the analysis of mean--field Ising systems
(and can be established, e.g., by further differentiation) the above
equation has only the trivial solution if $b\mu \leq 1$.  Since $b$
cannot be greater than one, the second statement has been proved -- in
fact whenever $\mu \leq 1$.
 \end{proof}

\noindent
The above result establishes, in a limited sense, the existence of a
phase transition.  Here we will sharpen this result by proving that
along the lines of fixed $b_{R}$, there is a single transition from $n
\equiv 0$ to $n > 0$.  This is an immediate corollary to the following
lemma which we state separately for future purposes.
\begin{lemma}
\label{Theorem Y}
Let $\Phi_{\mu}(b,n)$ denote the free energy function with $b_{R}$
fixed and $\mu$ (displayed) in $[0,\infty)$.  Then the minimizing
  $n(\mu)$, if unique, is a non--decreasing function of $\mu$.  More
  generally, if at various values of $\mu$, $\Phi_{\mu}$ has a
  minimizing set of $n$'s then, if $\mu^{\prime} > \mu$, the minimum
  of the minimizers at $\mu^{\prime}$ is greater than or equal to the
  maximum of the minimizers at $\mu$.  Thus, in general any possible
  ``choice'' of $n(\mu)$ is non--decreasing.
\end{lemma}
\begin{proof}
Let $\mu, \mu^{\prime} \in [0, \infty)$ with $\mu^{\prime} > \mu$
and let us denote by $(b^{\prime}, n^{\prime})$ a minimizing pair
for $\Phi_{\mu^{\prime}}$ and similarly for $(b,n)$ at $\mu$.  The
key observation is the meager $\mu$--dependence of the function
$\Phi_{\mu}$.  Indeed, $\Phi_{\mu^{\prime}}(x,y) = \Phi_{\mu}(x,y) -
\frac{1}{2}(\mu^{\prime} - \mu)y^{2}$.  We do this twice:
\begin{align}
\Phi_{\mu^{\prime}}(b^{\prime}, n^{\prime}) &= 
\Phi_{\mu}(b^{\prime},n^{\prime}) - \frac{1}{2}(\mu^{\prime} - \mu)[n^{\prime}]^{2}
\notag
\\
&\geq
\Phi_{\mu}(b,n) - \frac{1}{2}(\mu^{\prime} - \mu)[n^{\prime}]^{2}
\notag
\\
&=
\Phi_{\mu^{\prime}}(b,n) - \frac{1}{2}(\mu - \mu^{\prime})n^{2}
- \frac{1}{2}(\mu^{\prime} - \mu)[n^{\prime}]^{2}
\end{align}
leading to $\Phi_{\mu^{\prime}}(b^{\prime}, n^{\prime}) \geq
\Phi_{\mu^{\prime}}(b,n) + \frac{1}{2}(\mu^{\prime} - \mu)(n^{2} -
    [n^{\prime}]^{2})$.  
This necessarily implies that
    $[n^{\prime}]^{2} \geq n^{2}$ since otherwise, the previous
    inequality would be strict implying that $(b,n)$ would have been a
    ``better minimizer'' for $\Phi_{\mu^{\prime}}$ than $(b^{\prime},
    n^{\prime})$.
\end{proof}

\noindent
Using this result we may now show the following

\begin{corollary}
\label{UUY}
Consider the mean--field model defined by the free energy function
given in Eq.(\ref{OIQ}).  Then for each fixed $b_{R} \in (0,1]$, there
is a transitional value of $\mu$, denoted by $\mu_{T}(b_{R})$, such
that $n\equiv 0$ for $\mu < \mu_{T}$ and $n > 0$ for $\mu > \mu_{T}$.
\end{corollary}
\begin{proof}
This follows immediately from Proposition \ref{YTZ} and Lemma
\ref{Theorem Y} above.
\end{proof}
\noindent
Also of interest is the following:
\begin{corollary}
\label{JDD}
Consider the mean--field model defined by the free energy function
given in Eq.(\ref{OIQ}).  Let $n(\mu)$ denote any non--negative
function corresponding to a minimizing $n$ at parameter value $\mu$
(usually uniquely determined).  Then for $\mu \geq \mu_{T}$, the
function $n(\mu)$ is strictly increasing.
\end{corollary}
\begin{proof}
It is seen that \textit{if} $n(\mu_{T}) = 0$ then the statement of
this corollary is self-evident at $\mu =\mu_{T}$. For the rest of this
proof, we may simply assume that $\mu$ is such that $n(\mu) > 0$.
Suppose then that $\mu^{\prime} > \mu$ and that $n = n(\mu)$ is part
of the minimizing pair $(b(\mu), n(\mu))$ at parameter value $\mu$.
Suppose further that at $\mu^{\prime}$ the same $n$ is also part of a
minimizing pair.  Then we claim that the $b(\mu)$ is \textit{not} the
partner at $\mu^{\prime}$ since given $n^{\prime}$ -- purportedly
equal to $n$ -- then $b^{\prime}$ is uniquely determined by
Eq.(\ref{E:mf_n2}). Upon performing some algebraic manipulations the
latter reads $b^{\prime} = n^{\prime}/\tanh \mu^{\prime}n^{\prime}$.
Thus, the equality $n=n'$ would lead to
\begin{equation}
\nonumber
b^{\prime}  = \frac{n^{\prime}}{\tanh \mu^{\prime} n^{\prime}}
= 
\frac{n}{\tanh \mu^{\prime} n}
\neq
\frac{n}{\tanh \mu n}  = b
\end{equation}
so that explicitly $(b,n)$ cannot be a minimizer at parameter value
$\mu^{\prime}$.  Using the appropriate $b^{\prime} \neq b$, we would
have
\begin{equation}
\nonumber
F_{\text{MF}}(\mu^{\prime})  =
\Phi_{\mu^{\prime}}(n, b^{\prime})
= \Phi_{\mu}(n, b^{\prime}) - \frac{1}{2}(\mu^{\prime} - \mu)n^{2} 
\geq 
 \Phi_{\mu}(n, b) - \frac{1}{2}(\mu^{\prime} - \mu)n^{2} = \Phi_{\mu^{\prime}}(n, b)
\end{equation}
in contradiction with the fact that $(b,n)$ is not minimized for the
parameter value $\mu^{\prime}$.
\end{proof}

\subsection{A discontinuous Transition for $b_{R} < \frac{1}{3}$} 
The dividing point of $b_{R} = \frac{1}{3}$ along the LSP--curve $\mu = 1/b_{R}$ is apparent from the higher order terms in Eq.(\ref{KIT}).  Indeed, supposing $\mu = 1/b_{R}  +  \varepsilon$ we obtain, with the additional aid of Eq.(\ref{E:relNandDelta}),
\begin{equation}
\label{HUH}
\varepsilon n =  \frac{n^{3}}{2b_{R}^{3}}(b_{R} - \frac{1}{3}) + \dots
\end{equation}
For $b_{R} \geq \frac{1}{3}$, Eq.(\ref{HUH}) is consistent (and, as it
turns out correct) but in the case of $b_{R} < \frac{1}{3}$
this equation alone precludes the possibility of a continuous
transition.  Indeed since we cannot have $n^{2} < 0$, the only logical
consequence of Eq.\,(\ref{HUH}) is $n \equiv 0$ for $\mu \gtrsim
\mu_{S}(b_{R})$, i.e., the transition occurs later.  But the lower
order term \textit{insisted} that $\mu = \mu_{S}(b_{R})$ was the only
viable candidate for a continuous transition.  Thus: the transition
cannot be continuous and, at least for $b_{R} < \frac{1}{3}$, the
preliminary assumption that $n$ goes to zero continuously can no
longer be sustained. In particular, for $b_{R} < \frac{1}{3}$,
perturbative analysis will never be valid because the relevant
quantities will \textit{never} be small.

This leaves open the possibility of a transition at some
$\mu_{T}(b_{R})$ that is different than $\mu_{S}$.  We
shall show that $\mu_{T} < \mu_{S}$ as a direct consequence of the
following:
\begin{proposition}
\label{TEF}
Consider the mean--field model defined via the free energy function
given by Eq.(\ref{OIQ}).  Then, if $b_{R} < \frac{1}{3}$ at $\mu =
\mu_{S}(b_{R})$ the quantity $n$ is strictly positive.
\end{proposition}
\begin{proof}
We expand the free energy function $\Phi(n,b)$ -- with $G$ eliminated
via Eq.\,\ref{VSR} -- about $b = b_{R}$, $n = 0$ and along the curve
$\mu = \mu_{S}(b_{R})$.  The convenient variables are now chosen as $b
= b_{R}(1 + \Delta)$ and $n = b_{R}m$.  We first note that all odd
terms in $m$ must vanish. In addition, the term linear in $\Delta$
vanishes due to the stationarity of $\Phi$ along the curve $\mu =
\mu_{S}(b_{R})$ and, as it turns out, so does the term which is
quadratic in $m$.  This leaves us with
\begin{eqnarray}
\nonumber
\Phi(b, n)  =  \Phi(b_{R}, 0) +  \frac{1}{2}b_{R}
\left[
\frac{1}{6}m^{4}  +  \frac{1}{1-b_{R}}\Delta^{2}  -  m^{2}\Delta
\right]
+ \dots
\end{eqnarray}
Examining the quadratic form in the variables $m^{2}$ and $\Delta$,
the condition for a local minimum is that
\begin{equation}
\nonumber
\frac{1}{6}\frac{1}{1-b_{R}} > \frac{1}{4}, 
\end{equation}
\noindent
i.e., $b_{R} > \frac{1}{3}$.  
We return to $b_{R}\geq\frac{1}{3}$ in the next subsection.
Of current relevance is the fact that for $b_{R} < \frac{1}{3}$, the
curve $\mu = b_{R}^{-1}$ is of a saddle point nature.  This implies
that there \textit{is} a direction of decrease which, as is easily
seen, is optimized, in the physical direction, when $m^{2} = 3\Delta$. 
It is concluded that under the stated conditions, we can produce a
pair $(b, n)$ with $n^{2} > 0$ (and $b > b_{R}$) such that the free
energy for the non--trivial pair is lower; we just make the
corresponding objects small enough to withstand the higher order
corrections.  Thus the actual minimum also must occur for non--trivial
values of the $n$ observable.
\end{proof}

\noindent
We now have
\begin{theorem}
\label{UKW}
Consider the mean--field system defined by the free energy function as
given in Eq.(\ref{OIQ}).  Then for $b_{R} < \frac{1}{3}$, there is a
discontinuous transition at some positive $\mu_{T}(b_{R}) <
\mu_{S}(b_{R})$.
\end{theorem}
\begin{proof}
\noindent
That a transition occurs at \textit{some} $\mu_{T} > 0$ is the
statement of Corollary \ref{UUY}.  Moreover, Lemma \ref{Theorem Y} and
the above analysis implies $\mu_{T} \leq \mu_{S}$.  The discussion
prior to Proposition \ref{TEF} demonstrates that at $\mu = \mu_{T}$,
the quantity $n$ is already positive.  It only remains to show that
the inequality relating $\mu_{S}$ and $\mu_{T}$ is strict.
To this end, let us reimplement the heretofore unnecessary notation
for the full dependence of the free energies on parameters.  We have
learned that for $b_{R} < \frac{1}{3}$, there is an $n_{\star} > 0$
and a $b_{\star}$ (with $b_{\star} > b_{R}$) such that
\begin{equation}
\nonumber
F_{\text{MF}}(b_{R},\mu_{S})  =  
\Phi_{b_{R},\mu_{S}}(b_{\star}, n_{\star}) < 
\Phi_{b_{R},\mu_{S}}(b_{R}, 0).
\end{equation}
Invoking Lemma \ref{Theorem Y}, it is now
sufficient to show that there is a $\delta\mu > 0$ such that for some
nonzero $\tilde{n}$, and some $\tilde{b}$, the inequality
$\Phi_{b_{R},\mu_{S}-\delta\mu}(\tilde{b}, \tilde{n}) <
\Phi_{b_{R},\mu_{S} - \delta\mu}(b_{R}, 0)$ can be shown to hold.
Once again, the key is the simple dependence of the free energy
functions on the parameter $\mu$.  Indeed, using $n_{\star}$ and
$b_{\star}$ as trials, we obtain
$$
\Phi_{b_{R},\mu_{S}-\delta\mu}(n_{\star},b_{\star})  =  F^{\text{MF}}_{b_{R},\mu_{S}}
 + \frac{1}{2}[\delta\mu ]n_{\star}^{2}
$$ while $\Phi_{b_{R},\mu_{S} - \delta\mu}(b_{R}, 0) \equiv
 \Phi_{b_{R},\mu_{S}}(b_{R}, 0) < F^{\text{MF}}_{b_{R},\mu_{S}}$.
 Thus, the desired inequality will indeed hold for all $\delta\mu$
 sufficiently small.
\end{proof}

\subsection{A continuous transition for $b_{R} \geq \frac{1}{3}$}
The starting point in our analysis is to show that at the purported
critical curve, the quantity $n$ actually vanishes.
\begin{proposition}
\label{UYV}
For $b_{R} > \frac{1}{3}$ and $\mu = b_{R}^{-1} =: \mu_{S}$, the
unique solution to the mean--field equations is $n = 0$ with $b =
b_{R}$.  In particular, $\Phi_{b_{R},\mu_{S}}(b_{R}, 0) <
\Phi_{b_{R},\mu_{S}}(b, n)$ for any $(b,n) \neq (b_{R}, 0)$.
\end{proposition}
\begin{proof}
Assuming $n > 0$ the agent fraction $b$ can be eliminated in favor of the ratio
$$
\theta := \frac{n}{b}.
$$ 
Note that while this is the same substitution as before, here it is
$b$ rather than $n$ that is being eliminated.  Notwithstanding,
$\theta$ still satisfies $0 < \theta \leq 1$.  In these variables,
the mean--field equations, Eq.(\ref{E:mf_b2}) and Eq.(\ref{E:mf_n2})
respectively become
\begin{equation}
\label{ASA}
n = \frac{R\theta}{R + \sqrt{1-\theta^{2}}}
\end{equation}
\begin{equation}
\label{BSB}
n = b_{R}\text{Arctanh}\hspace{.05 cm} \theta
\end{equation}

\noindent
where in the above, $R := b_{R}/(1 - b_{R})$.  Let us now
define $\ell(\theta)$ as
\begin{equation}
\nonumber
\ell(\theta) := \frac{1}{b_{R}}
\hspace{.1 cm} \frac{R\theta}{R + \sqrt{1-\theta^{2}}} = \frac{(1 +
 R)\theta}{R + \sqrt{1-\theta^{2}}} = \frac{(1 + R)\theta}{R +Q}.
\end{equation}
where $Q = Q(\theta) :=\sqrt{1 - \theta^{2}}$.
To prove the current proposition we need to show that for all $\theta > 0$,
\begin{equation}
\nonumber
\text{Arctanh}\hspace{.05 cm} \theta > \ell(\theta) 
\end{equation}
demonstrating that there cannot be a non--trivial solution to the
mean--field equations under the conditions stated.  Note that for $0 <
\theta \ll 1$ the desired inequality can be explicitly demonstrated.
In general, it is sufficient to show, for $0 < \theta \leq1$, that
$\ell^{\prime}(\theta) < 1/(1 - \theta^{2})$, i.e., in the $Q$
variable that
\begin{equation}
\nonumber
\frac{1}{Q^{2}}  >  \frac{(1+R)(R + Q + (1-Q^{2})/Q)}{(R+Q)^{2}}.
\end{equation}
\noindent
Although both sides diverge as $Q\to 0$ the divergence on the left
hand side is clearly stronger so we actually only need consider $Q >
0$ limiting us to $Q \in (0,1)$.  After some manipulation, the
inequality we need to prove is equivalent to
\begin{equation}
\nonumber
(R+Q)^{2}  >  (1 + R)(RQ^{2} + Q)  =  (1 + R)(RQ^{2} + Q^{2}) + (1+R)(Q - Q^{2}).
\end{equation}
That is, we now wish to show
\begin{equation}
\nonumber
R (R + 2Q + RQ)(1 - Q) > (1+R)Q(1 - Q).
\end{equation}
Since $Q \neq 1$, the above is equivalent to 
\begin{equation}
\nonumber
R^{2} + QR + R^{2}Q > Q.
\end{equation}
Finally, since also $Q < 1$ it is enough to show that $2R^{2}Q^{2} +RQ \geq
Q$ i.e., that $2R^{2} + R \geq 1$ which occurs for $R \geq
\frac{1}{2}$. This corresponds to $b_{R} \geq \frac{1}{3}$.
\end{proof}
\noindent
We can finally show
\begin{theorem}
Consider the mean--field GI--system defined by the free energy
function given in Eq.(\ref{OIQ}).  Then, for $b_{R} \geq \frac{1}{3}$,
as a function of $\mu$ with $b_{R}$ fixed, there is a continuous
transition at $\mu = \mu_{S} = 1/b_{R}$.  I.e., $n(\mu) \equiv 0$ for
$\mu < \mu_{S}$ and $n(\mu) > 0$ for the $n$--component of any
minimizing pair $(n(\mu), b(\mu))$ while, if $\mu \downarrow \mu_{S}$,
it is found that $n(\mu)\downarrow 0$.
\end{theorem}
\begin{proof}
We will marshal the facts at our disposal and then proceed in a more
abstract vein than has been the case in the more recent of our
arguments. In what is to follow, $n(\mu)$ and the corresponding
$b(\mu)$ is, once again, notation for a minimizing pair without any
claims to uniqueness.
By the preceding proposition, we know that at $\mu = \mu_{S}$, the
quantity $n(\mu)$ is unambiguous and vanishes for $\mu < \mu_{S}$ by
Lemma \ref{Theorem Y}. Conversely, for $\mu > \mu_{S}$ we may write,
adhering to the notation in the proof of Theorem \ref{UKW}, our usual
expression:
\begin{equation}
\nonumber
\Phi_{b_{R},\mu}(b,n)  =  \Phi_{b_{R},\mu_{S}}(b,n) - \frac{1}{2}(\mu - \mu_{S})n^{2}.
\end{equation}
\noindent
For $n^{2} \propto b - b_{R} \ll1$ from Proposition \ref{TEF}, we know
that the quantity $\Phi_{b_{R}, \mu_{S}}(b, n)$ agrees with
$\Phi_{b_{R}, \mu_{S}}(b_{R}, 0)$ up to quartic order in $n$.  Thus
allowing $n^{2} \ll 1$ with $n^{2}(\mu - \mu_{s}) \gg n^{4},
(b-b_{R})^{2}$ we find a non--zero $n$ corresponding to a free energy
lower than that of $\Phi_{b_{R}, \mu_{S}}(b_{R}, 0)$.  Therefore,
again by Lemma \ref{Theorem Y}, we have $n(\mu) > 0$ for all $\mu >
\mu_{S}$.  It remains to establish that $n\downarrow 0$ as $\mu
\downarrow \mu_{S}$. Note, that along any decreasing sequence of
$\mu$'s the corresponding possible $n$'s must be monotone by Corollary
\ref{JDD} -- or even Lemma \ref{Theorem Y} -- and hence
$n\downarrow 0$ as $\mu \downarrow \mu_{S}$. 
Now let us suppose otherwise: that for some
sequence of $\mu$'s decreasing to $\mu_{S}$ there is an associated
sequence of minimizers, $(b(\mu), n(\mu))$ that has $n(\mu) \downarrow
n_{\star} > 0$.  Let $b_{\star}$ denote the associated limit for the
$b(\mu)$ along a further subsequence if necessary.  Since
\begin{equation}
\nonumber
\Phi_{b_{R}, \mu}(b(\mu), n(\mu))  < \Phi_{b_{R}, \mu}(b_{R}, 0) \equiv 
\Phi_{b_{R}, \mu_{S}}(b_{R}, 0)
\end{equation} 
we would have, by continuity, $\Phi_{b_{R}, \mu_{S}}(b_{\star},
n_{\star}) \leq \Phi_{b_{R}, \mu_{S}}(b_{R}, 0)$ indicating that
\textit{at} $\mu = \mu_{S}$, there is a minimizer with positive
magnetization in contradiction with Proposition \ref{UYV} above.  It
follows that, under the stated condition $b_{R} \geq \frac{1}{3}$, the
limit of $n(\mu)$ is zero as $\mu \to \mu_{S}$ while it vanishes below
and is positive above.  By this (and any other) criterion, the
transition at $\mu_{S}$ is continuous. This completes the proof.
\end{proof}

\section{Discussion} \label{S:discussion}

In this work, we have formulated a lattice model for gang
territoriality where red and blue gang agents interact solely through
graffiti markings. Using a contour argument, we showed that a phase
transition occurs between a well mixed, ``high-temperature'' phase and
an ordered, ``low-temperature'' one as the coupling parameter $J$
between gang members and graffiti becomes stronger while the graffiti
evaporation parameter $\lambda$ decreases. In the mean field limit of
all--to--all lattice site couplings, we can also identify the
tricritical point in phase space that distinguishes the occurrence of
a continuous phase transition from a first order one. We find this
point to be located at $b_R = 1/3$ which corresponds, in terms of the
original variables of the problem, to the gang proclivity 
term $\alpha
= - 2 \log 2 $.  In particular, for $b_R \geq 1/3$ the phase
transition is continuous and occurs at $\mu = 1 / b_R$. Thus, in the
mean-field limit, for fixed $\alpha \geq -2 \log 2$ the ordered ``low
temperature" phase arises for $J^2 > \lambda / (e^{-\alpha} +2)$, and
the ``high temperature'' one is attained on the other side of this
inequality. The transition between the two occurs in a continuous
manner across the $J^2 = \lambda/ (e^{-\alpha}+ 2)$ locus.  In the
opposite case of $b_R < 1/3$ (or $\alpha < -2 \log 2$) the phase
transition is discontinuous. Here, we also are able to prove that the
transition between high and low temperature phases occurs not at $\mu
= 1/ b_R $, but rather along the $\mu = \mu_T < 1 /b_R$ curve, so that
the phase change occurs earlier in $J$ and along a separatrix $J^2 =
J_c^2 < \lambda / (e^{-\alpha} + 2)$.

In the context of gang--graffiti interactions, we may identify the low
temperature, clustered phase as pertaining to a high level of
antagonism between between rival gangs, where segregation leads to
conflict along boundaries. Vice versa, the high temperature, well mixed
configuration can be interpreted as a peaceful state,
where despite different affiliations, gang members share the same
turf. Our mean field results indicate that the confrontational state
is surely attained, whether in a continuous or first order manner, for
$J^2 > \lambda / (e^{-\alpha} + 2)$, which represents high
gang-graffiti territoriality $J$, low external intervention in
graffiti removal $\lambda$ and high proclivity $\alpha$ for
individuals to become gang members.  Gang clustering can be avoided
by intervening in all three directions: by externally eliminating
graffiti ($\lambda$), but also, from a deeper sociological point of
view, by decreasing the lure of graffiti tags
or of joining gangs in the
first place ($J,\alpha$).
The emergence of a (continuous or discontinuous) phase transition
shows that it is possible to obtain segregation in a lattice model
without invoking direct agent--to--agent coupling; it is certain that adding
such coupling terms to the Hamiltonian would allow for even more
favorable segregation conditions.

Although our work was conceived within the context of gang
interactions, the proposed model Hamiltonian and the tools used are
general enough that our fundamental results may be applicable to
several other contexts where territoriality is played out through
markings and not through direct contact between players.  Many
animals, among which wolves, foxes and coyotes, are known to
scent--mark their territories as a way of warning intruders of their
presence and to exchange internal communication \cite{Levin}.  At
times, buffer zones can originate between distinct animal clusters
where prey species, such as deer or moose, may thrive \cite{White}.
Insects, such as beetles and bees, are also known to avoid previously
marked locations as a way to optimize foraging patterns.  Similarly to
the role of gang graffiti markings, foreign scents lead ``others'' to
retreat from already occupied turf or visited patches. Our work also
applies to these contexts. Although some stochastic treatments have
been recently presented \cite{Giuggioli}, classical ecological studies
of territoriality are usually carried out via reaction--diffusion
equations where focal points such as dens, burrows or nests are often
included \cite{Lewis, Murray}, leading to segregation. Within this
work on the other hand -- whether first order or continuous -- agent
clustering is a natural consequence of a probabilistic treatment
without the need to include any anchoring sites. Finally, we are able
to connect local microscopic parameters -- $J, K, \lambda, \alpha$ --
to the emergence of large scale territorial patterns, be they gang
clusters or animal groupings.

\vspace{.5cm}


\noindent \textbf{Acknowledgments: }
This work was supported by NSF grants DMS--0968309 (A.B. and L.C.), DMS--0805486 (L.C.), DMS--0719642 and DMS--1021850 (M.R.D.) and by ARO grants  W911NF--11--1--0332 (A.B. and M.R.D.) and W911NF--10--1--0472 (A.B.).


\vspace{.5cm}
\appendix
\section{}
\noindent
Here we present a brief description of the informal simulations
mentioned in Section\,\ref{S:MFHamiltonian}.  We consider a 100
$\times$ 100 square lattice with periodic boundary conditions
initialized at $t=0$ so that each site is populated with either red or
blue agents, or a mixture of both.  We assume a random distribution of
$10^5$ blue and $10^5$ red agents and do not impose any restriction on
the number of individuals on each site, so that multiple agents can
occupy the same location at any given time.  Initial conditions are
completed by assuming that at $t=0$ there is no graffiti present. At
each time step of the simulation agents leave their graffiti on-site
with a probability $p_m$ that depends on the current graffiti level.
In particular, $p_m=0.1$ if the site is not marked by the opposite
gang, and $p_m=1$ otherwise.  
The agent then
moves to any of its nearest neighbor sites $j$ with probability
$e^ {-g_j} / \sum_{j} e^{g_j}$, where $g_j$ is the amount of
the opposite gang's graffiti at location $j$.  Agents will thus
preferentially relocate to nearest neighbor sites tagged by the least
amount of the opposite gang's graffiti.  Finally, at each site,
graffiti is removed according to a probability $p_g$.  Similarly to
the number of agents, we impose no restriction on the amount of
graffiti at each site.

\begin{table}[t]
\begin{center}
\begin{tabular}{|c|c|c|}
\hline
& & \\
gang \hspace{0.73cm} graffiti & gang \hspace{0.73cm} graffiti &
gang \hspace{0.73cm} graffiti \\ 
\scalebox{0.15}{\includegraphics{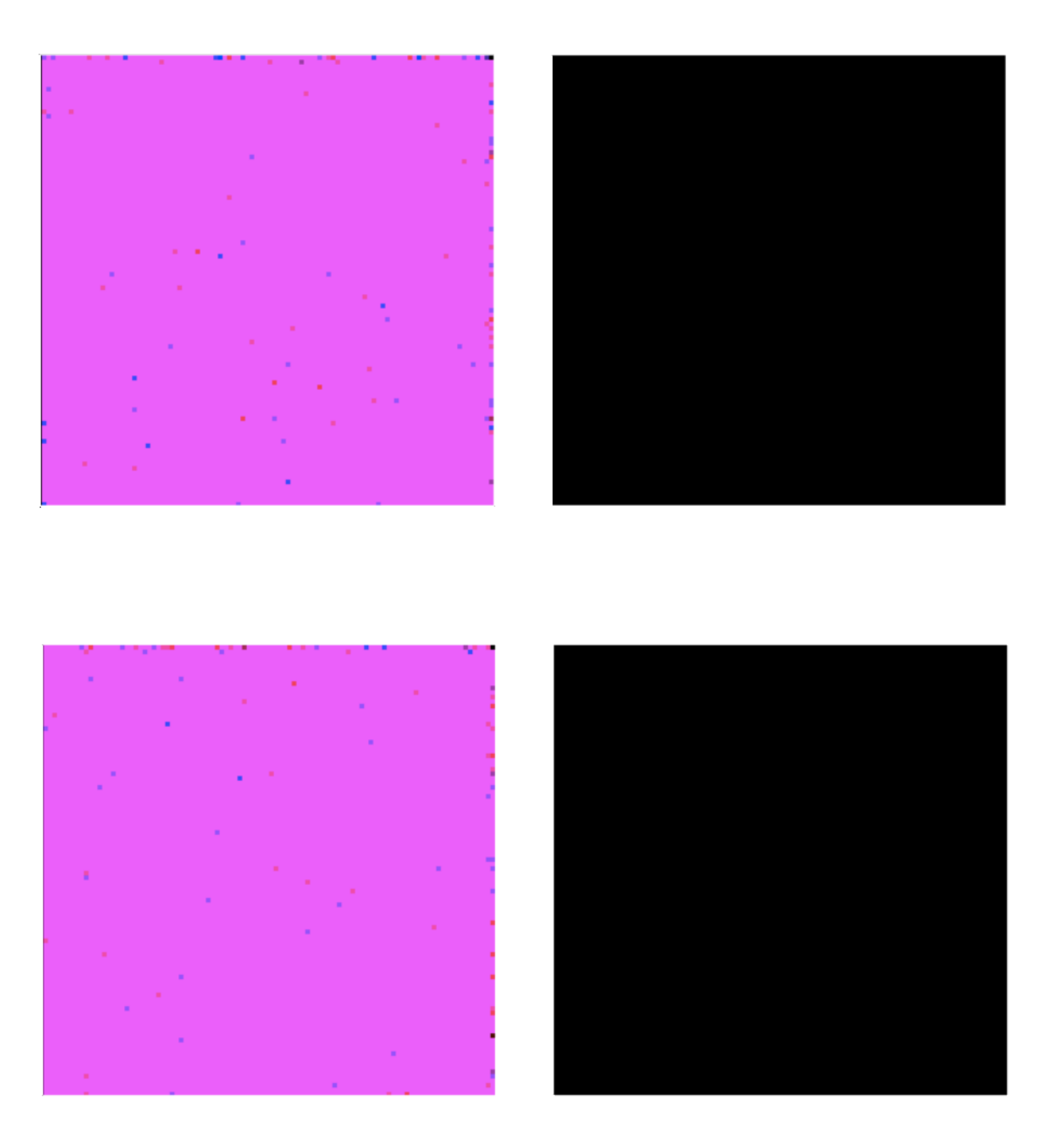}}&
\scalebox{0.15}{\includegraphics{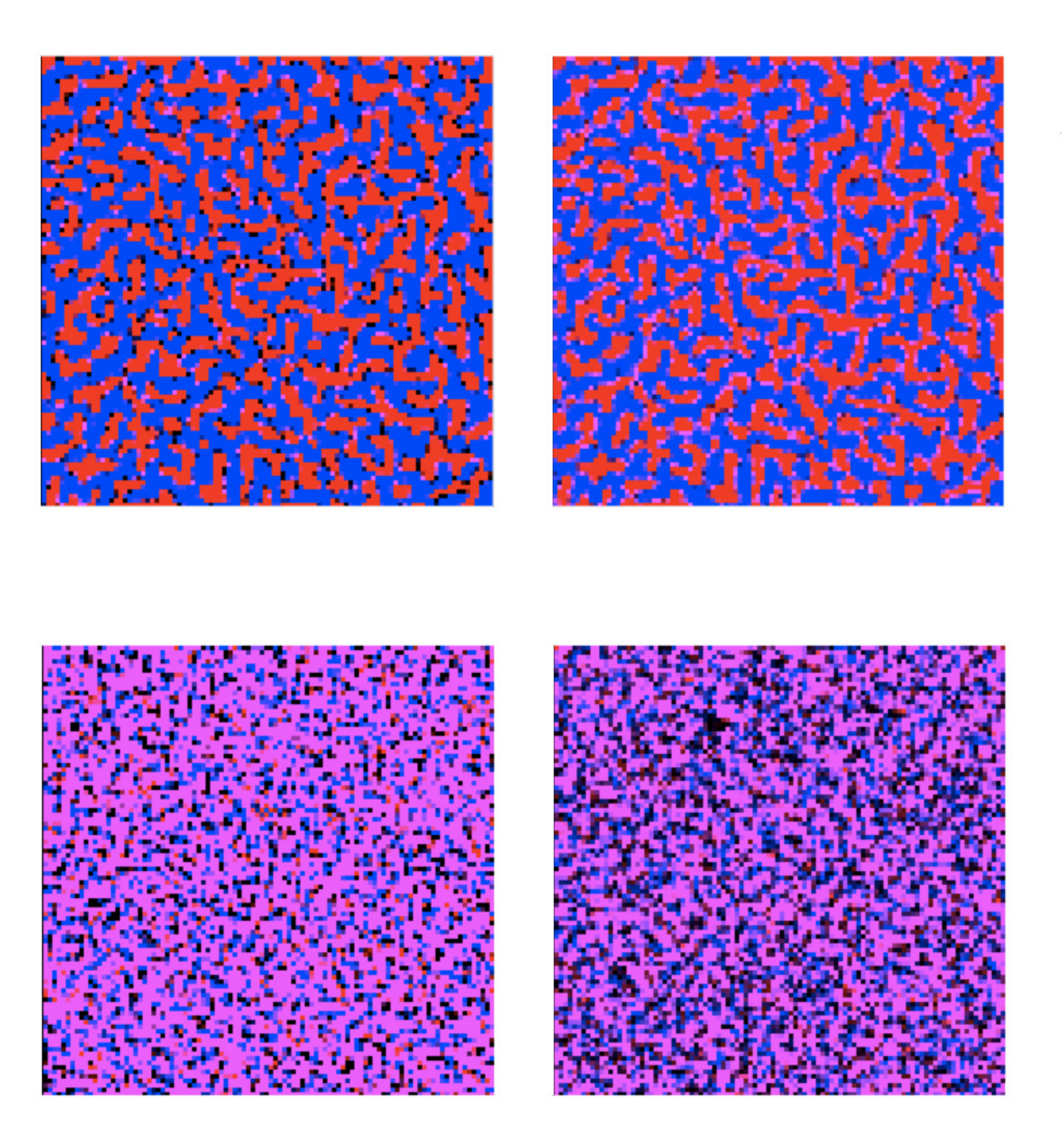}}&
\scalebox{0.15}{\includegraphics{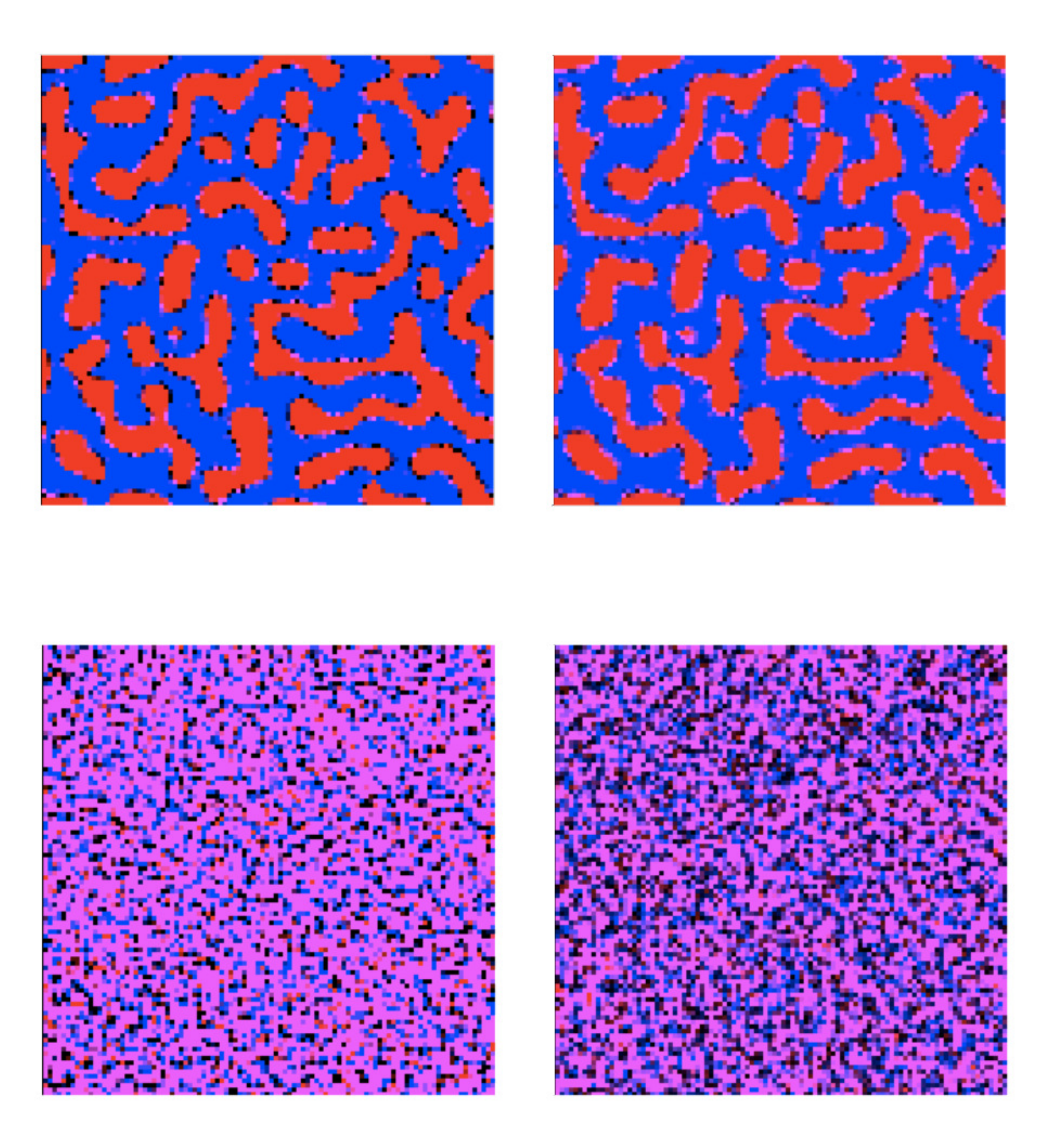}}\\ 
(a) t=0 & (b) t=100
& (c) t= 1000 \\  & & \\ \hline
& & \\
 gang \hspace{0.73cm} graffiti & gang \hspace{0.73cm} graffiti &
 gang \hspace{0.73cm} graffiti \\
 \scalebox{0.15}{\includegraphics{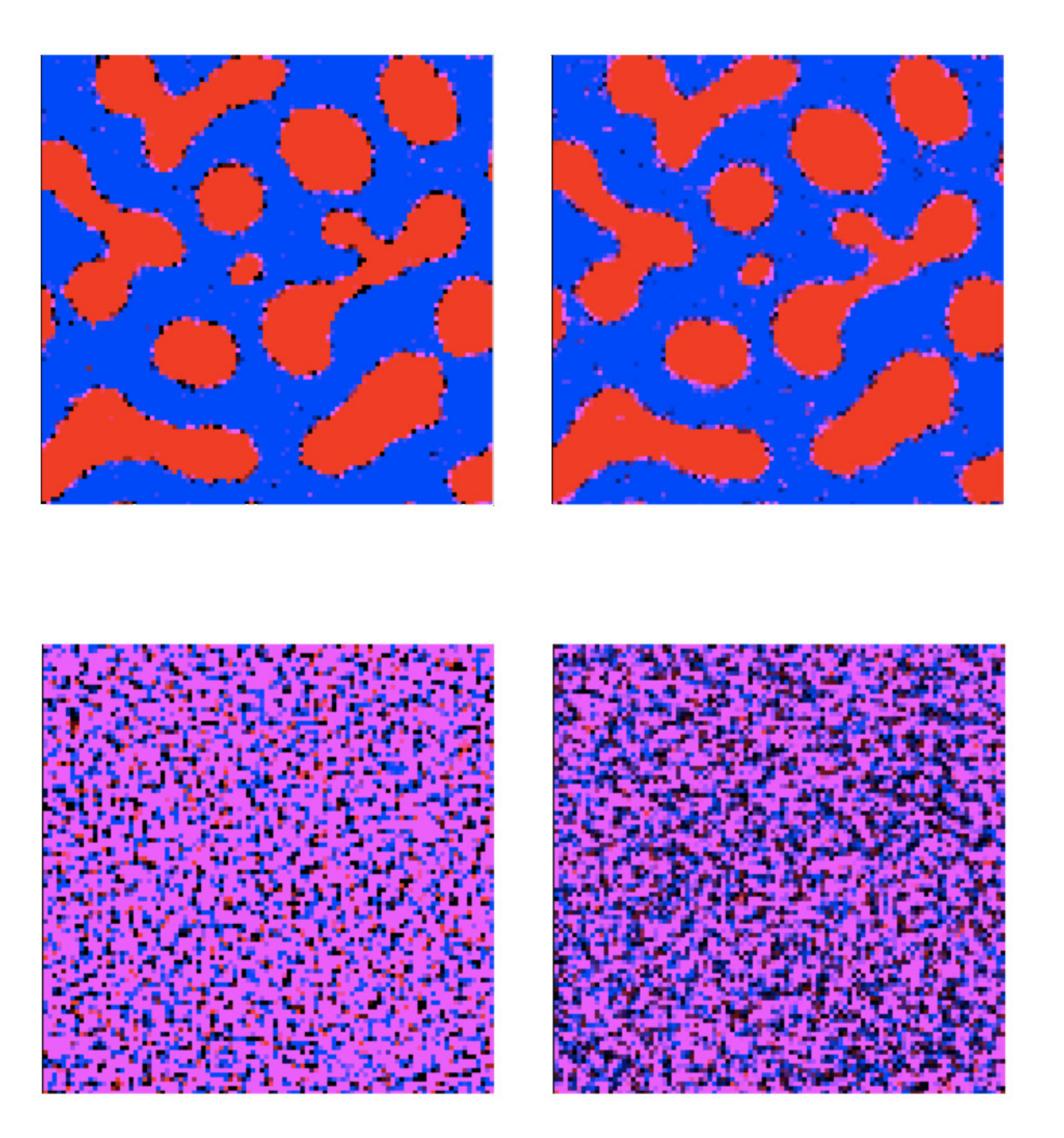}}&
 \scalebox{0.15}{\includegraphics{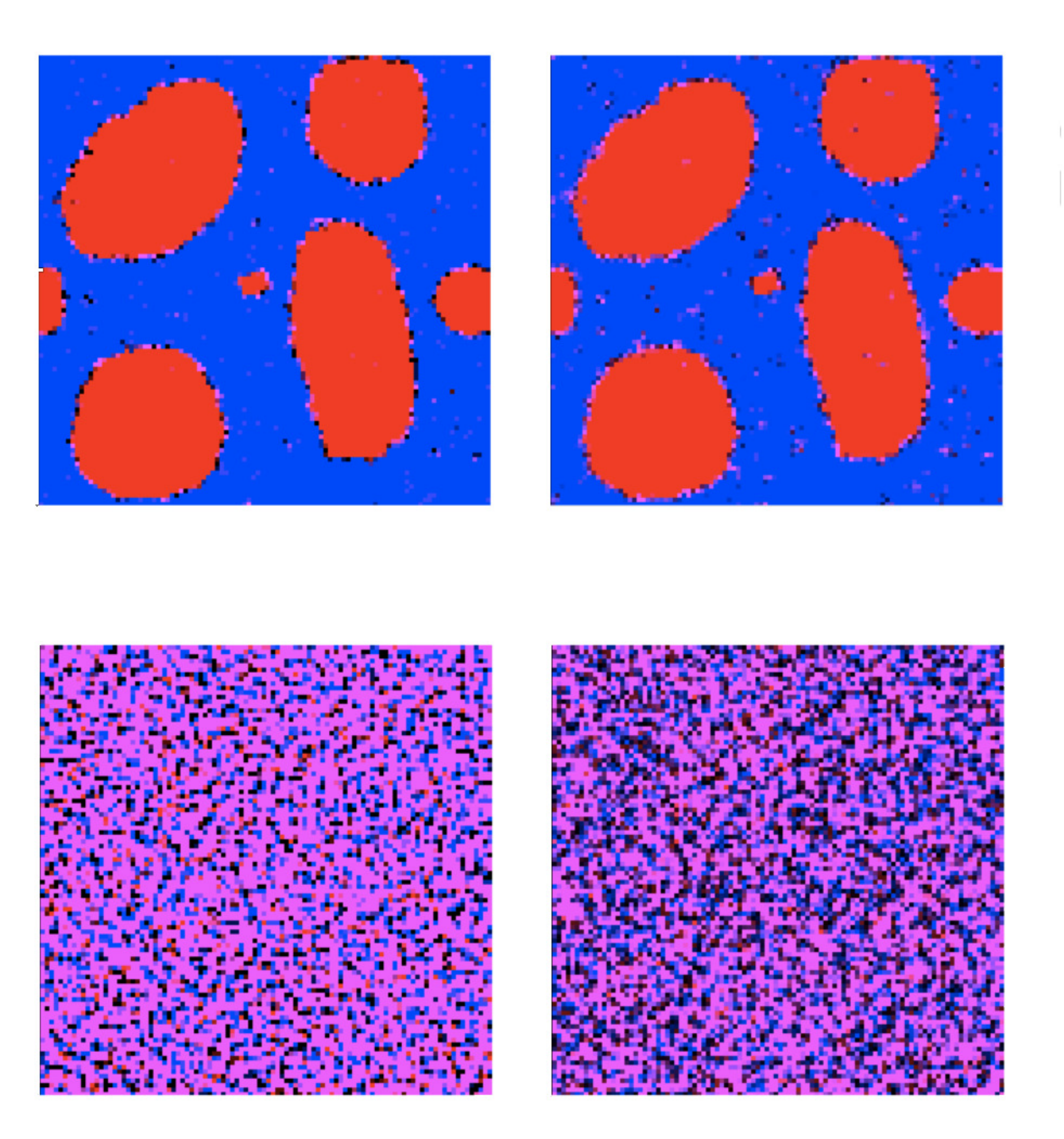}}&
\scalebox{0.15}{\includegraphics{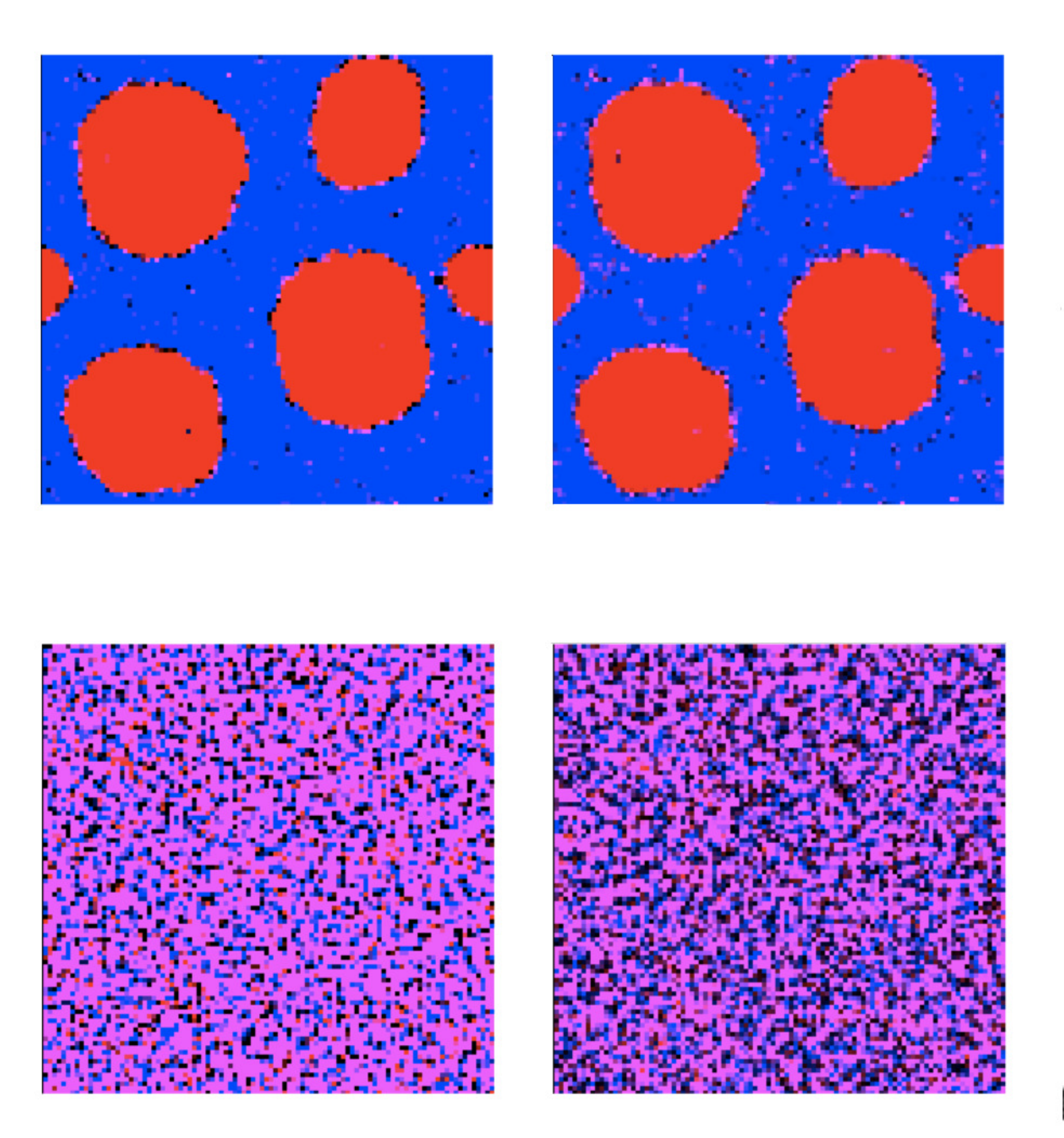}}\\ 
(d) t= 10000 &
(e) t= 100000 & (f) t= 150000 \\ &  & \\ \hline
\end{tabular}
\vspace{0.2cm}
\caption{Snapshots of a Monte Carlo simulation of gang dynamics on a
  100 $\times$ 100 square lattice with periodic boundary conditions.
  For each image in the sequence, the upper left panel represents gang
  agent populations while the upper right panel is the corresponding
  graffiti distribution. Iteration time is measured in arbitrary units.
  At $t=0$, $10^5$ red and $10^5$ blue gang members are placed at
  random on the lattice with possible overlaps. Magenta indicates a
  mixture and black indicates a void in gang agents or graffiti.
  Agents tag their sites with probability $p_m=0.1$ if the site is not
  marked by the opposite gang's graffiti and with probability $p_m=1$
  otherwise.  In the upper panels of the table entries, $p_g=0.25$ so
  that at each time step, graffiti will persist with a $75 \%$
  possibility. The lower panels, where $p_g=0.75$, mirror the upper
  ones but with a much lower graffiti persistence, of $25 \%$. Note
  the different outcomes of the simulations at long times: when
  graffiti is allowed to persist longer, segregation occurs with the
  formation of islands of red and blue gangs. In this work, just as in
  our current model, there is no direct interaction between gang
  members, underlying the importance of the graffiti field as an
  indirect coupling between agents.}
\label{tab:FIG1}
\end{center}
\end{table}

While the simulation rules described here are similar in spirit to the
model we analyze in this work, they do not directly lead to the
Hamiltonian in Eq.\,\ref{E:Hamiltonian}.  These informal simulations
however provided us with a playing ground to investigate any phase
transitions that may take place upon varying relevant parameters, such
as the graffiti removal probability $p_g$.  For example, in Table\,
\ref{tab:FIG1}, we track the dynamic progression of two sets of
parameters.  In the top row we set $p_g = 0.25$ so that 75$\%$ of the
graffiti is retained at each iteration, while in the lower one we set
$p_g = 0.75$ so that only 25$\%$ of the graffiti is kept. In Table
\,\ref{tab:FIG1} the left and right hand side plots show agent and
graffiti distributions, respectively.  Red and blue pixels indicate
site occupied by respective gang agents, black pixels represents no
agents, and magenta shades indicate coexistence of both red and blue
agents.  Just as in the main body of this paper, we do not include any
direct coupling between red and blue agents who interact only via the
graffiti field.  Similarly to what we later found in the main
analysis, the degree of persistence of the graffiti field -- which can
be related to $\lambda$ in the Hamiltonian in Eq.\,\ref{E:Hamiltonian}
-- yields different qualitative behaviors and, if sufficiently large,
may lead to aggregation patterns with distinct red and blue
phases. The emergence of separate clusters from these simulations
motivated the more extensive study 
presented in this work.

\end{document}